\documentclass[11pt,letterpaper]{article}

\usepackage[margin=1in]{geometry}
\usepackage{theorem,latexsym,graphicx,amssymb}
\usepackage{amsmath,enumerate}
\usepackage{bm,bbm}
\usepackage{wrapfig}
\usepackage{subfigure}
\usepackage{float}
\usepackage{epsfig}
\usepackage{xspace}
\usepackage{paralist}
\usepackage{enumerate}
\usepackage{cases}
\usepackage{caption}
\usepackage{algpseudocode}
\usepackage[ruled,linesnumbered,vlined]{algorithm2e}

\usepackage{xcolor}
\usepackage{complexity}
\usepackage{multicol}
\usepackage{graphicx}
\usepackage{thm-restate}

\newenvironment{proof}{{\bf Proof:  }}{\hfill\rule{2mm}{2mm}\vspace*{5pt}}

\numberwithin{figure}{section}
\numberwithin{equation}{section}

\newtheorem{theorem}{Theorem}[section]
\newtheorem{lemma}{Lemma}[section]
\newtheorem{claim}{Claim}[section]

\newcommand{\bbE}{{\mathbb{E}}}
\newcommand{\bbR}{\mathbb{R}}
\newcommand{\ALG}{\textsf{ALG}}
\newcommand{\OPT}{\textsf{OPT}}

\usepackage{url}

\usepackage[colorlinks=true, allcolors=blue]{hyperref}

\title{Prophet Inequality on I.I.D. Distributions:
Beating \texorpdfstring{$1-1/e$}{} with a Single Query}

\author{
	Bo Li\thanks{The Hong Kong Polytechnic University. {\texttt{comp-bo.li@polyu.edu.hk}}}
	\and Xiaowei Wu \thanks{University of Macau. {\texttt{xiaoweiwu@um.edu.mo}}} 
	\and Yutong Wu \thanks{The University of Texas at Austin. {\texttt{yutong.wu@utexas.edu}}}
}
\date{}

\begin{document}

\begin{titlepage}
	\thispagestyle{empty}
	\maketitle
	
	\begin{abstract}
		In this work, we study the single-choice prophet inequality problem, where a gambler faces a sequence of~$n$ online i.i.d. random variables drawn from an unknown distribution. When a variable reveals its value, the gambler needs to decide irrevocably whether or not to accept the value.  The goal is to maximize the competitive ratio between the expected gain of the gambler and that of the maximum variable.  It is shown by Correa et al. \cite{ec/CorreaDFS19} that when the distribution is unknown or only $o(n)$ uniform samples from the distribution are given, the best an algorithm can do is $1/e$-competitive. In contrast, when the distribution is known~\cite{sigecom/CorreaFHOV17} or $\Omega(n)$ uniform samples are given~\cite{innovations/RubinsteinWW20}, the optimal competitive ratio of 0.7451 can be achieved. In this paper, we study a new model in which the algorithm has access to an oracle that answers quantile queries about the distribution and investigate to what extent we can use a small number of queries to achieve good competitive ratios. We first use the answers from the queries to implement the threshold-based algorithms and show that with two thresholds our algorithm achieves a competitive ratio of $0.6786$. Motivated by the two-threshold algorithm, we design the observe-and-accept algorithm that requires only a single query. 
    This algorithm sets a threshold in the first phase by making a query and uses the maximum realization from the first phase as the threshold for the second phase. It can be viewed as a natural combination of the single-threshold algorithm and the algorithm for the secretary problem. By properly choosing the quantile to query and the break-point between the two phases, we achieve a competitive ratio of $0.6718$, beating the benchmark of $1-1/e$.
	\end{abstract}
\end{titlepage}

\section{Introduction}

Prophet inequality is one the of most widely studied problems in optimal stopping theory. 
In the classic setting, a gambler faces a sequence of online non-negative random variables $x_1,\ldots,x_n$ independently drawn from distributions $D_1,\ldots,D_n$. 
The realizations of these variables represent potential rewards given to the gambler.
After observing the realization of a variable, the gambler needs to make an irrevocable decision on whether to accept the reward and leave the game, or reject it and continue with the next variable.
All rejected values cannot be collected anymore.
A prophet knows all the realizations beforehand and thus can always select the highest reward. 
The expected reward of the prophet is then defined as $\OPT=\bbE_{x_1,\ldots,x_n} [ \max_i\{x_i\} ]$.
The gambler's goal is to choose a stopping rule so that her expected reward, which we refer to as $\ALG$, is as close to that of the prophet as possible.
The performance is measured by the \emph{competitive ratio}, which is defined as the worst case of $\ALG/\OPT$ over all possible distributions.
The prophet inequality problem has received increasingly more attention due to its connections with Bayesian mechanism design,  especially the posted price mechanism, where buyers arrive online and the seller provides each buyer with a take-it-or-leave-it offer \cite{aaai/HajiaghayiKS07,stoc/ChawlaHMS10,journals/sigecom/Lucier17}.

It is proved in \cite{krengel1977semiamarts,krengel1978semiamarts} and \cite{samuel1984comparison} that there exists a stopping rule such that the gambler's expected reward is at least $0.5\cdot\OPT$, and
this is the best possible competitive ratio when the distributions are distinct and the order of variables is adversarial. 
However, when the distributions are identical, i.e., $D_i = D$ for all $i$, much better competitive ratios can be achieved. 
Hill and Kertz\cite{hill1982comparisons} proved that there is an algorithm that achieves a competitive ratio of $1-1/e\approx 0.6321$ while Kertz \cite{kertz1986stop} showed that no algorithms can do better than $0.7451$. 
The best-known competitive ratio remained $1-1/e$ until Abolhassani et al. \cite{stoc/AbolhassaniEEHK17} improved it to $0.738$.
Later, Correa et al. \cite{sigecom/CorreaFHOV17} proposed a {\em blind quantile strategy} with a tight competitive ratio of $0.7451$.
Informally speaking, a blind quantile strategy defines a sequence of increasing probabilities $p_1 < \cdots < p_n$ (which depend on $n$ and the distribution $D$, but are independent of the realizations), so that the acceptance probability for each variable $x_i$ equals $p_i$.
This is done by setting the $i$-th threshold $\theta_i$ such that $\Pr_{x \sim D}[x > \theta_i] = p_i$ and accepting the $i$-th variable if $x_i>\theta_i$.
However, computing the probabilities $p_1,\ldots,p_n$ that give the optimal ratio is highly non-trivial and requires complete information of the distribution $D$.
As pointed out in a survey paper \cite{sigecom/CorreaFHOV18}, an interesting research problem is to investigate the amount of information needed to achieve a good competitive ratio: 
\begin{quote}
  {\em How much knowledge of the distributions is required to achieve a good competitive ratio?}  
\end{quote}

To follow up this question, a line of recent works studies the prophet inequality problem on unknown distributions, most of which focus on the setting where uniform random samples from the distribution are given~\cite{geb/AzarKW19,ec/CorreaDFS19,innovations/RubinsteinWW20}.
Particularly, it is shown in \cite{innovations/RubinsteinWW20} that $O(n)$ samples are sufficient to achieve a competitive ratio arbitrarily close to $0.7451$, and in \cite{ec/CorreaDFS19} that no algorithm can use $o(n)$ samples to ensure a result better than $1/e(\approx 0.368)$-competitive.
However, in some applications obtaining random samples can be difficult.
Consider the following analogy of prophet inequality: a seller sells a product with a posted price to a sequence of buyers whose values are independently drawn from a distribution. 
A buyer purchases the product if and only if the price is lower than her private value.
In practice, the seller often does not have an accurate knowledge of the value distribution or random samples from this distribution. 
The seller can, however, observe the selling records to draw a relationship between a historic price and the fraction of buyers that have accepted the price, which is precisely a {\em quantile} of the value distribution.  
Since one price is usually offered over months and even years, which means that the number of historic prices is often small, such quantile information is also scarce. 
Therefore, a natural question is how a seller can take advantage of the limited quantile information to set up a price that maximizes their expected revenue. 

Motivated by the above discussion, in this work, we propose a new information model in which the algorithm has access to an oracle that answers quantile queries on the unknown distribution. 
Formally, given a quantile $q\in [0,1]$, the oracle returns a value $v$ such that $\Pr_{x\sim D}[x>v] = q$.
In other words, the oracle answers queries in the form of ``\emph{what price should be set to generate a success probability of $q$ when buyers' values are independently drawn from the distribution $D$?''}  
We aim to understand how much reward can be guaranteed using only a few queries and construct such algorithms.
Our model aligns well with the study of query complexity in Bayesian auction design -- a closely related area to prophet inequality.
In particular, \cite{conf/sigecom/AllouahBB21} considered the revenue maximization problem with a single value-quantile pair, and showed that, for example, if we are given the value at quantile $0.5$, we are able to guarantee $85\%$ of the optimal revenue when there is a single buyer and the full distribution is known. 
As far as we know, the query model has not been studied in the context of prophet inequality. 

Besides algorithms that use little knowledge, \emph{simple} algorithms are often much preferred due to their easy implementation in real-world scenarios.
For example, in the {\em secretary problem} \cite{gilbert1966recognizing}, a widely used simple strategy is to discard the first $1/e$ fraction of the variables and select the first variable whose realization is greater than all the discarded ones.
We call this strategy the  {\em observe-then-accept} algorithm.
It was proved in \cite{ec/CorreaDFS19} that this algorithm is $1/e$-competitive and is optimal when the distribution is unknown or $o(n)$ uniform samples are given. 
For the prophet inequality problem, the {\em single-threshold} algorithm 
 simply  
uses a fixed threshold for all variables and accepts the first variable whose realization exceeds the threshold.
It is shown in \cite{conf/soda/EhsaniHKS18} that by setting an appropriate threshold, the single-threshold algorithm ensures a competitive ratio of $1-1/e$, which is the best possible ratio using just one threshold\footnote{We remark that their result holds for the more general setting of the prophet secretary problem. When the distributions are discrete, randomization is required to achieve the $1-1/e$ competitive ratio.}.
In this work, we propose two simple algorithms for the query model that achieve good competitive ratios.
We first use quantile queries to obtain thresholds and
extend the single-threshold algorithm to a simple {\em multi-threshold algorithm}. 
The proposed multi-threshold algorithm naturally does better than $(1-1/e)$-competitive and we give a general framework to derive its competitive ratio.  
Besides using multiple thresholds, we prove that beating $1-1/e$ can also be done by using a single query via our {\em observe-and-accept algorithm}.   
Indeed, an interesting take-home message from this work is that a simple 
combination of the observe-then-accept algorithm and the single-threshold algorithm performs strictly better than $(1-1/e)$-competitive.

A review of more related works can be found in Appendix \ref{sec:related-works}.

\subsection{Our Contribution}

We study the prophet inequality problem with unknown i.i.d. distributions and propose simple algorithms that use a few queries to achieve competitive ratios strictly better than $1-1/e$. 
Our contribution is summarized as follows.

\paragraph{Two-threshold Algorithm.}
Our first result is to show that a simple two-threshold algorithm achieves a competitive ratio strictly better than $1-1/e$, the optimal ratio of single-threshold algorithms \cite{conf/soda/EhsaniHKS18}.
Specifically, the algorithm fixes two quantiles $q_1, q_2$ (which define two thresholds $\theta_1,\theta_2$) and divides the time horizon $\{1,\ldots,n\}$ into two phases.
The threshold $\theta_1$ (resp. $\theta_2$) is then used to decide the acceptance of variables in the first (resp. second) phase.
Our main technical contribution is a careful formulation of the lower bounds on $\ALG$ and upper bounds on $\OPT$ in terms of the parameters, which then become the objective and constraints for a minimization linear program (LP) whose optimal solution provides a lower bound on the competitive ratio.
By fixing appropriate parameters and solving the LP, we show that the competitive ratio of the algorithm is at least $0.6786$.\footnote{Our analysis extends naturally to the multi-threshold algorithm. In Appendix~\ref{sec:3-or-more-queries}, we show that better competitive ratios can be achieved with more queries and thresholds.}
We also complement our positive result with an upper bound on the competitive ratio: we show that any two-threshold algorithm cannot do better than $0.7081$-competitive.

\paragraph{Observe-and-accept Algorithms.}
Our second result, which is the main contribution of this work, is showing an algorithm using a single query that can do strictly better than $(1-1/e)$-competitive.
Most interestingly, we show that a natural combination of the single-threshold algorithm and the observe-then-accept algorithm 
gives one such algorithm, with a competitive ratio of at least $0.6718$.
The algorithm works in a similar way as the two-threshold algorithm in that it partitions the time horizon into two phases and uses the query result from the oracle as a threshold to decide the acceptance of variables in the first phase.
In the second phase, however, instead of making a second query, the algorithm uses the maximum realization of variables in the first phase as the threshold, given that no variables are accepted during the first phase.
We refer to this algorithm as \emph{observe-and-accept}, as it is similar to the algorithm for the secretary problem but in addition, allows acceptance of variables in the observation phase.
As before, we lower bound the competitive ratio by bounding $\ALG$ from below and  $\OPT$ from above to formulate an LP.
However, different from the analysis for the two-threshold algorithm, the main technical challenge is to lower bound $\ALG$ when the second threshold is a random variable that depends on the realizations of variables in the first phase, i.e., the algorithm is not a blind strategy algorithm.
To overcome this main difficulty, we extend the  LP by further dividing the domain of the second threshold into continuous segments, whose boundary points become additional variables of the LP.
Actually, given any quantile, we can compute the optimal separation point of the two phases and obtain the corresponding optimal reward by running the observe-and-accept algorithm.
Via carefully designing the query and the separation point of phases, our main result is a $0.6718$-competitive observe-and-accept algorithm.
To complement this result, we provide hard instances for which the observe-and-accept algorithm performs no better than 0.6921-competitive.

\paragraph{Robustness of the Quantile to Query.}
If we are not allowed to choose an arbitrary quantile but are instead  given a specific one, we show in Figure~\ref{fig:c_ratio}
some examples when the given quantile is $1-c/n$ for some constant $c$; a detailed setup will be introduced in Section \ref{sec:single:FactorRevealingLP}.
Recall that in our motivating example, the thresholds are the posted prices of an item for a sequence of $n$ buyers.
We can view $c/n$ as $c$ out of $n$ buyers have private values higher than the price.
As $c$ increases, the corresponding price decreases.
We can observe that as long as $c$ is between $0.5$ and $1.0$,
the expected reward from observe-and-accept is always greater than ($1-1/e$). 
When $c$ is between $0.1$ and $2.5$, the quantile information is still beneficial in the sense that the observe-and-accept algorithm has a competitive ratio better than $1/e$ -- the best we can do when no information is given.

\begin{figure}[ht]
    \centering
    \includegraphics[width=0.6\textwidth]{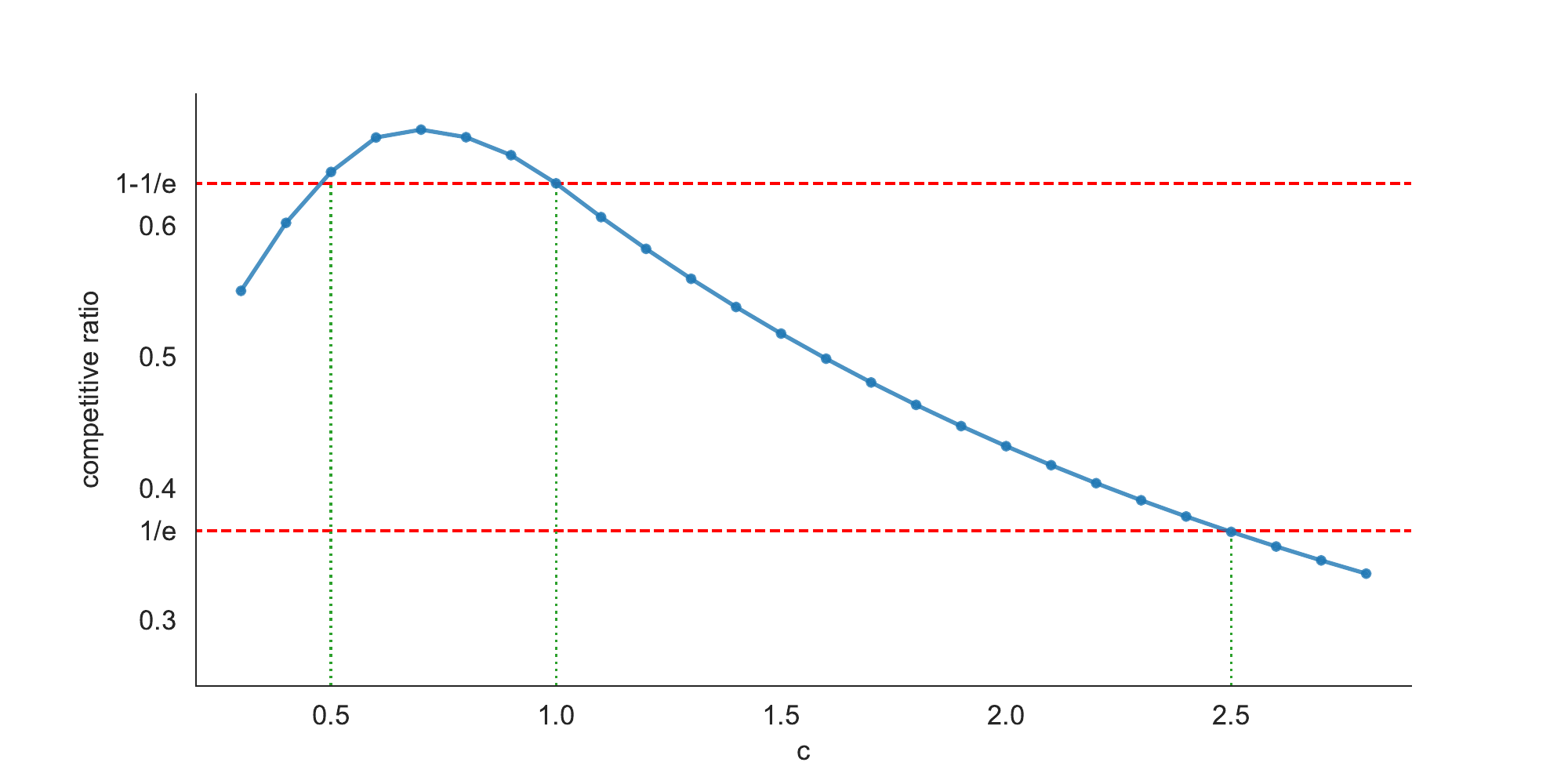}
    \caption{Competitive ratios for different quantiles at $1-c/n$.}
    \label{fig:c_ratio}
\end{figure}

\paragraph{Bayesian Auctions with Queries.}
Besides prophet inequality, queries have also been studied in Bayesian auctions when the seller does not know the prior distributions.
Similar to our motivation,  Allouah et al. \cite{conf/sigecom/AllouahBB21} considered the single-item single-buyer revenue maximization problem with one value-quantile pair.
They showed that, for example, 
if we are given the value at quantile $0.01$, 
$51\%$ of the optimal revenue is guaranteed when the full distribution is known,
and if we are given the value at quantile $0.5$, 
the guarantee is improved to $85\%$.
Chen et al.\cite{journals/ai/ChenLLL22} studied a more general setting of multi-item multi-buyer auctions and gave an almost tight number of queries to ensure a constant fraction of the optimal revenue. 
Hu et al. \cite{conf/sigecom/Hu0SW21} explored the middle ground between random samples and queries by allowing the algorithm to specify a quantile interval and sample from the prior distribution restricted to the interval.
In a different flavor from the quantile query oracle we study in this work, Leme et al. \cite{corr/abs-2111-03158} defined {\em pricing queries} in that when the oracle is given a price, it generates a random sample and returns the sign of whether the sample is above the given price or not. 
Both query-related settings are of independent interest to be studied in prophet inequality.

\subsection{Organization of the Paper}
The rest of this paper is organized as follows.
We introduce the necessary notations and some basic properties in Section~\ref{sec:prelim}.
In Section~\ref{sec:warmup}, we review the analysis for proving the $1-1/e$ competitive ratio of the single-threshold algorithm (Section~\ref{ssec:single-thres}), introduce the two-threshold algorithm (Section~\ref{ssec:2-queries}), and show that it achieves a competitive ratio of at least $0.6786$ using a factor revealing LP (Section~\ref{ssec:comp-ratio-2-thresholds}).
Built on top of the algorithm and analysis from Section~\ref{sec:warmup},  we propose in Section~\ref{sec:observe-and-accept} the observe-and-accept algorithm that uses a single query to achieve a competitive ratio of at least $0.6718$.
Finally, we conclude our results and propose some open problems in Section~\ref{sec:open-problem}.

\section{Preliminaries}
\label{sec:prelim}

In our problem, there are $n$ non-negative random variables $\{x_1,\ldots,x_n\}$ that are independently and identically drawn from 
an unknown distribution $D$. 
We slightly abuse the notation and use $x_i's$ to denote both random variables and their realizations.
The realizations of the variables are revealed to the algorithm one by one.
When {the value of a variable} is revealed, the algorithm needs to make an irrevocable decision on whether to \emph{accept} the value and stop the algorithm, or \emph{reject} this variable and proceed to the next one.
The objective is to pick a realization that is as large as possible.
For any integer $t$, we use $[t]$ to denote $\{1,2,\ldots,t\}$.
Throughout this paper, we use $\tau \in [n+1]$ to denote the stopping time of the algorithm.
If the algorithm accepts variable $x_i$, then $\tau=i$; if the algorithm does not accept any variable, we set $\tau = n+1$. 
We let $\ALG$ denote the expected gain of the algorithm.
We compare the performance of the algorithm with a prophet, who sees all the realizations of variables before making a decision.
In other words, the benchmark is
\begin{equation*}
     \textstyle \OPT=\bbE_{(x_1,\ldots,x_n) \sim D^n}\left[\max_{i\in[n]} \{x_i\} \right].
\end{equation*}

The \emph{competitive ratio} of the algorithm is defined as the minimum of ${\textsf{ALG}}/{\textsf{OPT}}$ over all the distributions.
In this work, we consider the setting where the distribution is unknown but the algorithm can make queries on the cumulative distribution function (CDF) of the distribution.

\paragraph{Query.}
We denote the CDF of the distribution by $F(\theta)=\Pr_{x\sim D}[x < \theta]$ and the survival function by $G(\theta)=\Pr_{x\sim D}[x \ge \theta]$.
Note that for every distribution, the CDF $F(\cdot): \mathbb{R} \rightarrow [0,1]$ is non-decreasing.
In this work, we assume the distribution does not contain mass points. 
That is, $\Pr_{x\sim D}(x = t)=0$ for all $t$.
When the distribution $D$ is known, this assumption is w.l.o.g. because we can add an arbitrarily small random noise to the distribution.
On the other hand, when the distribution $D$ is unknown, the assumption is apparent\footnote{If there are mass points in the distribution, the query model will not be well-defined since for some quantile query there might not exist any value with this particular quantile. We will then have to adopt some fixed rule to return a nearby value as in Equation \eqref{eq:query:round}. In this way, the distributions would be distorted, and we demonstrate in Appendix~\ref{appendix:justification} that this distortion can be arbitrarily large in some cases.} in order to achieve a good competitive ratio with few queries.
For distributions without mass points, the CDF $F(\cdot)$ is strictly increasing.
Given any quantile query at $q\in [0,1]$, the distribution quantile oracle returns the corresponding value $v(q)$ such that
\begin{equation}\label{eq:query:round}
    v(q) := \inf_{\theta\in \bbR} \left\{ F(\theta) \ge q \right\} = F^{-1}(q).
\end{equation}
Therefore, if $v(q)$ is used as a threshold to decide the acceptances of variables,  each variable is accepted with probability $1-q$. 
We note that $v(\cdot)$ is also strictly increasing.

For any number $t$, we use $(t)^+$ to denote $\max\{ t,0 \}$.
Since we often make queries at $q=1-c/n$ for some constant $c$ (i.e., probability to accept is $c/n$), for convenience we define 
\begin{equation*}
    \Delta(c):= n\cdot\bbE_{x\sim D}[(x-v(1-\tfrac{c}{n}))^+].
\end{equation*}

\paragraph{Zero Queries.}
We briefly discuss the case when we cannot make any query to the distribution oracle. 
The problem then degenerates to the model where the algorithm has no information on the distribution \cite{ec/CorreaDFS19}.
The classic algorithm for the secretary problem \cite{gilbert1966recognizing} achieves the optimal competitive ratio of $1/e$, which also applies to our setting.
Formally, we observe the first $n/e$ variables without accepting any of them; for the following variables, we accept the first variable whose realization exceeds the maximum of the previous ones. 
Clearly, the algorithm retains the $1/e$ competitive ratio since it can be equivalently described as independently drawing $n$ realizations from the distribution and giving them a random arrival order.  
Interestingly, Correa et al. \cite{ec/CorreaDFS19} showed that this is the best possible ratio for all algorithms that have no information about the distribution. 

\begin{theorem}[Zero-Query]
There is a $1/e$-competitive algorithm that makes no queries for the prophet inequality problem on unknown i.i.d. distributions.
Moreover, the competitive ratio is optimal.
\end{theorem}

\section{Warm-up Analysis and Techniques} \label{sec:warmup}

We first provide a warm-up analysis for single-threshold algorithms in Section~\ref{ssec:single-thres}, which is the least a single query can do by using the query result as a threshold to accept the first variable whose realization exceeds the threshold.
We show that the competitive ratio of the algorithm is $1-1/e$.
We then show in Section~\ref{ssec:2-queries} that by using two thresholds that are obtained by making two queries, we can do strictly better than $1-1/e$.
Our main technique is based on the theory of factor revealing LPs.
Both the algorithm and the technique from Section~\ref{ssec:2-queries} will be important building blocks towards proving our main result in Section~\ref{sec:observe-and-accept}.

\subsection{The Single-Threshold Algorithm}\label{ssec:single-thres}

Single-threshold algorithms have been investigated by  \cite{conf/soda/EhsaniHKS18} and  \cite{soda/CorreaSZ19}, where a $(1-1/e)$ competitive ratio was proved. 
Moreover, the competitive ratio is optimal for single-threshold algorithms~\cite{conf/soda/EhsaniHKS18}.
For completeness, we give the formal proof here, which will also serve as a warm-up for our proofs in later sections. 

\begin{theorem}\label{thm:warmup:singleT}
There exists a $(1-1/e)$-competitive algorithm for the prophet inequality problem with unknown i.i.d. distributions that makes a single query to the distribution oracle.
\end{theorem}
\begin{proof}
The algorithm works as follows:
\begin{enumerate}
    \item Let $\theta = v(1-1/n)$. In other words, we have $\Pr_{x\sim D}[x\ge \theta] = 1/n$.
    \item For the realization of each variable $x_i$, where $i=1,2,\ldots,n$, accept $x_i$ and terminate if $x_i\geq \theta$; otherwise reject it and observe the next variable, if any.
\end{enumerate}

We refer to $\theta$ as the \emph{threshold} of the algorithm.
Note that $F(\theta) = 1-1/n$ and $G(\theta) = 1/n$. 
We show that the competitive ratio for the above algorithm is at least $1-1/e$. 
We prove this by deriving an upper bound for $\OPT$ and a lower bound for $\ALG$. 
Let $x^*=\max_{i}\{x_i\}$ denote the maximum of the $n$ i.i.d. random variables. 
Then $\OPT=\bbE[x^*]$. 
Given threshold $\theta$, 
we have 
\begin{align*}
    \OPT=\bbE[x^*]&=\bbE[x^*\mid x^*< \theta]\cdot\Pr[x^*< \theta]+\bbE[x^*\mid x^*\ge \theta]\cdot \Pr[x^*\ge \theta]\\
    &=\bbE[x^*\mid x^*< \theta]\cdot\Pr[x^*< \theta]+(\theta+\bbE[x^*-\theta\mid x^*\ge \theta])\cdot\Pr[x^*\ge \theta]\\
    &=\bbE[x^*\mid x^*< \theta]\cdot\Pr[x^*< \theta]+\theta\cdot\Pr[x^*\ge \theta]+\bbE[x^*-\theta\mid x^*\ge \theta]\cdot\Pr[x^*\ge \theta].
\end{align*}
Note that the conditional expectation $\bbE[x^*\mid x^*<\theta]$ is at most $\theta$. Moreover, 
\begin{equation}\label{eq:split-arg}
    \bbE[(x^*-\theta)^+] = 0\cdot\Pr[x^* \le \theta] +\bbE[x^*-\theta\mid x^* \ge \theta]\cdot \Pr[x^* \ge \theta].
\end{equation}
Thus we obtain the following upper bound on $\OPT$,
\begin{align*}
    \OPT \le \theta \cdot \Pr[x^* < \theta]+\theta\cdot\Pr[x^*\ge \theta]+\bbE[(x^*-\theta)^+] = \theta+\bbE[(x^*-\theta)^+].
\end{align*}
Since $x^*$ is the maximum among all $x_i$'s and $x_i$'s are i.i.d., we have 
$$\bbE[(x^*-\theta)^+] \le \sum_{i=1}^n \bbE[(x_i-\theta)^+]=n\cdot\bbE[(x-\theta)^+]=\Delta(1).$$
Overall, a valid upper bound for $\OPT$ is given by 
\begin{equation}\label{eq:opt-upper}
    \OPT \le \theta+\Delta(1).
\end{equation}

Recall $\tau \in [n+1]$ is the stopping time of the algorithm.
The expected gain of the algorithm can then be written as 
\begin{equation}\label{eq:alg-decomp}
    \ALG = \sum_{i=1}^n\Pr[\tau=i]\cdot \bbE[x_i\mid \tau=i].
\end{equation}
The algorithm  accepts the $i$-th variable if and only if the first $i-1$ variables are smaller than the threshold $\theta$ and $x_i$ is at least $\theta$.
In other words, we have
\begin{equation}\label{eq:alg-prob}
    \Pr[\tau=i]=F(\theta)^{i-1}\cdot G(\theta).
\end{equation}
Note that $\bbE[x_i\mid\tau=i] = \bbE[x_i\mid x_i \geq \theta]$.
Conditioned on $x_i \geq \theta$, the algorithm receives $\theta$ and the expected difference between $x_i$ and $\theta$. 
That is, 
\begin{equation}\label{eq:alg-gain}
    \bbE[x_i\mid\tau=i] = 
    \theta + \bbE[x_i-\theta\mid x_i\ge \theta]= \theta+\frac{\bbE[(x_i -\theta)^+]}{G(\theta)}=\theta+\frac{\Delta(1)/n}{G(\theta)},
\end{equation}
where the second last equality is derived from a similar argument to (\ref{eq:split-arg}).
By plugging Equations (\ref{eq:alg-prob}) and (\ref{eq:alg-gain}) into Equation~\eqref{eq:alg-decomp}, we have 
\begin{align*}
    \ALG &=\sum_{i=1}^n F(\theta)^{i-1}\cdot G(\theta)\cdot \left(\theta+\frac{\Delta(1)/n}{G(\theta)}\right)=\sum_{i=1}^n\left(1-\frac{1}{n}\right)^{i-1}\left(\frac{\theta}{n}+\frac{\Delta(1)}{n}\right)\\
    &=\frac{1-(1-\frac{1}{n})^n}{1/n}\cdot\frac{1}{n}\cdot(\theta+\Delta(1))\ge \left(1-(1-1/n)^n\right)\cdot\OPT.
\end{align*}
Since $(1-1/n)^n\le 1/e$ for all $n\in \mathbb{N}_+$, we have $\ALG \ge (1-1/e)\cdot\OPT$.
\end{proof}

\subsection{The Two-Threshold Algorithm}
\label{ssec:2-queries}

Next, we show that one can beat the competitive ratio of $1-1/e$ with a two-threshold algorithm\footnote{The algorithm and analysis framework naturally extends  to three or more thresholds, and the details are deferred to Section~\ref{sec:3-or-more-queries} in the appendix.}.
We first briefly explain why the performance of single-threshold algorithms is upper bounded by $1-1/e$, and how to get around this upper bound using multiple thresholds.
We observe that any single-threshold algorithm faces the following trade-off: the threshold cannot be too high such that the algorithm terminates with no acceptance, and the threshold cannot be too low such that the algorithm terminates too early without observing sufficiently many variables.
Recall that $\tau \in [n+1]$ is the stopping time.
The trade-off is reflected by  $\Pr[\tau \leq n]$ (the probability of accepting some variable) and $\frac{1}{n}\sum_{i=1}^n \Pr[\tau \geq i]$ (the average probability of ``seeing" a variable). 
In fact, the competitive ratio of any threshold-based algorithm is upper  bounded  by these two terms.
Suppose the threshold is $\theta = v(1-c/n)$ for some $c$.
It can be shown that $\Pr[\tau \leq n] \geq 1-1/e$ only if $c\geq 1$ while $\frac{1}{n}\sum_{i=1}^n \Pr[\tau \geq i] \geq 1-1/e$ only if $c\leq 1$.
Therefore, one has to set $c = 1$ to achieve the competitive ratio of $1-1/e$, which is optimal. 
Now assume that we can use two thresholds: $\theta_1 = v(1-c_1/n)$ and $\theta_2 = v(1-c_2/n)$, for some $c_1 < c_2$.
In order to beat $1-1/e$, we need both $\frac{1}{n}\sum_{i=1}^n \Pr[\tau \geq i]$ and $\Pr[\tau \leq n]$ to exceed $1-1/e$.
To ensure $\Pr[\tau \leq n] > 1-1/e$, we must set $c_2 > 1$, as otherwise, we have $c_1 < c_2 \leq 1$ and the algorithm is strictly worse than the single-threshold algorithm with $c=1$.
Similarly, to make sure that $\frac{1}{n}\sum_{i=1}^n \Pr[\tau \geq i] > 1-1/e$, we must set $c_1 < 1$.
By fixing an appropriate division point $\rho\in (0,1)$ and using threshold $\theta_1$ for variables $\{x_i\}_{i\leq \rho n}$, $\theta_2$ for variables $\{x_i\}_{i > \rho n}$, we show that the algorithm does strictly better than $(1-1/e)$-competitive.
The main challenge lies in characterizing $\ALG$ and $\OPT$ in terms of  the thresholds and the division points, and choosing the parameters that optimize the competitive ratio.

\paragraph{Algorithm.}
The algorithm has three parameters $c_1, c_2$ and $\rho$, where $c_1 < c_2$ controls the quantiles of the two queries and $\rho$ decides the fraction of variables on which we use the first threshold.
Specifically, let thresholds be $\theta_1 := v(1-c_1/n)$ and $\theta_1 := v(1-c_2/n)$ via querying the oracle.
We use threshold $\theta_1$ to decide the acceptance of variable $x_i$ for $i \leq \rho\cdot n$ and threshold $\theta_2$ for $x_i$'s with $i \geq \rho \cdot n +1$.
For sufficiently large $n$, we can assume w.l.o.g. that $\rho \cdot n$ is an integer.
We describe the algorithm formally in Algorithm~\ref{alg:2-threshold}.

\begin{algorithm}
\caption{Two-Threshold Algorithm}\label{alg:2-threshold}
	\textbf{Input:} {{Number of items} $n$, parameters $c_1, c_2$ and $\rho$.}
 
    \textbf{Output:} A variable selected by the gambler.

    Set $\theta_{1} := v(1-c_{1}/n)$ by making a query to the oracle; 

    Set $\theta_{2} := v(1-c_{2}/n)$ by making another query to the oracle; 

    \For{$i=1,2,\ldots,\rho\cdot n$}{
		\If{$x_i \geq \theta_1$}{ 
		    \textbf{return} $x_i$ and \textbf{terminate};\qquad\qquad {\color{gray}// phase 1}
        }
    }
    
    \For{$i=\rho\cdot n+1,\rho\cdot n+2,\ldots,n$}{
		\If{$x_i \geq \theta_2$}{ 
		    \textbf{return} $x_i$ and \textbf{terminate};\qquad\qquad {\color{gray}// phase 2}
        }
    }
\end{algorithm}

We refer to the first for-loop (line 5 - 7) as the first phase of the algorithm and the second for-loop (line 8 - 10) as the second phase.
Since $c_1 < c_2$, we have $\theta_1 \geq \theta_2$. 
Our main goal is to characterize $\ALG$ and $\OPT$ using a factor revealing linear program (LP), whose optimal objective lower bounds the competitive ratio of the two-threshold algorithm.

\paragraph{Remark.}
The technique of \emph{factor revealing LPs} was first introduced by \cite{DBLP:journals/jacm/JainMMSV03} and further extended by  \cite{DBLP:conf/stoc/MahdianY11}. 
Generally speaking, a family of LPs is called factor revealing if the infimum of the optimal values for these LPs serves as a lower bound on the competitive ratio for the original maximization problem. 
In our case, we construct factor revealing LPs for the adversary of the gambler. 
After the gambler has established an algorithm to select variables, the goal for the adversary is to find a distribution of the variables such that the gambler's gain is minimized. 
We approximate the nonlinear optimization problem faced by the adversary with an LP whose objective will be a lower bound for the competitive ratio of the gambler. 
The objective function and the constraints of the factor revealing LP are formulated using  inequalities similar to but more advanced than the ones derived in Section~\ref{ssec:single-thres}.
We then leverage the power of LP solvers to find the best algorithm for the gambler such that the optimal value of the family of LPs is maximized. 
In the following subsections, we will show that by carefully choosing the parameters,
the optimal value of the LPs (and thus the competitive ratio) for the two-threshold algorithm is at least $0.6786$.

\subsection{Lower Bounding $\ALG$}\label{ssec:alg-lb-2-thresholds}

We derive a lower bound for the expected gain of the algorithm in the following lemma. 

\begin{lemma} \label{lem:ALG-lb-2-thresholds}
Let $\epsilon > 0$ be an arbitrarily small constant.
For sufficiently large $n$, we have
\begin{equation}\label{eq:alg-lb-2-thresholds}
   \ALG \ge (1-\epsilon)\cdot \left( 
   (1-e^{-\rho\cdot c_{1}}) \cdot \left(\theta_{1} +\frac{\Delta(c_{1})}{c_{1}}\right) 
   + e^{- \rho \cdot c_1} \cdot (1-e^{-(1-\rho) \cdot c_{2}}) \cdot \left(\theta_{2} +\frac{\Delta(c_{2})}{c_{2}}\right)
   \right).
\end{equation}
\end{lemma}
\begin{proof}
    Let $\tau \in [n+1]$ denote the stopping time of the algorithm.
    Specifically, if $\tau \leq n$,  the algorithm terminates after accepting variable $x_\tau$; if $\tau = n+1$, the algorithm terminates without accepting any variable.
    Then we can write the expected gain of the algorithm as
    \begin{equation*}
        \ALG = \sum_{i=1}^{n} \left( \Pr[\tau = i] \cdot \bbE[x_i \mid \tau = i] \right).
    \end{equation*}
    
    Recall that the algorithm has two phases.
    \begin{enumerate}
        \item For all $i\leq \rho\cdot n$, the algorithm accepts variable $x_i$ if and only if $x_j < \theta_1$ and $x_i\geq \theta_1$.
        \item For all $i\geq \rho\cdot n+1$, the algorithm accepts the $x_i$ if and only if (i) $x_j < \theta_1$ for all $j \leq \rho\cdot n$, (ii) $x_j < \theta_2$ for all $j\in [\rho\cdot n + 1, i-1]$, and (iii) $x_i\geq \theta_2$.
    \end{enumerate}
    
    Therefore, we have 
    \begin{equation*}
        \Pr[\tau = i] = 
        \begin{cases}
            F(\theta_1)^{i-1}\cdot G(\theta_1), \qquad\qquad\qquad\qquad \text{if } i\leq \rho\cdot n; \\
            F(\theta_1)^{\rho\cdot n}\cdot F(\theta_2)^{i-1-\rho\cdot n}\cdot G(\theta_2), \qquad \text{if } i\geq \rho\cdot n + 1.
        \end{cases}
    \end{equation*}
    
    Similar to Equation~\eqref{eq:alg-gain}, the expected gain of the algorithm given that it accepts $x_i$ is
    \begin{equation*}
        \bbE[x_i \mid \tau = i]  = 
        \begin{cases}
            \theta_{1} +\frac{\Delta(c_{1})/n}{G(\theta_{1})}, \qquad \text{if } i\leq \rho\cdot n; \\
            \theta_{2} +\frac{\Delta(c_{2})/n}{G(\theta_{2})}, \qquad \text{if } i\geq \rho\cdot n + 1.
        \end{cases}
    \end{equation*}

    Given that $F(\theta_1) = 1-\frac{c_1}{n}$ and $F(\theta_2) = 1-\frac{c_2}{n}$, the expected gain of the algorithm is 
    \begin{align*}
        \ALG 
        &= \sum_{i=1}^{\rho \cdot n} F(\theta_{1})^{i-1}\cdot G(\theta_{1}) \cdot \left(\theta_{1} +\frac{\Delta(c_{1})/n}{G(\theta_{1})}\right) \\
        & \qquad + \sum_{i=\rho \cdot n + 1}^{n} F(\theta_1)^{\rho \cdot n} \cdot  F(\theta_{2})^{i-1-\rho\cdot n}\cdot G(\theta_{2}) \cdot \left(\theta_{2} +\frac{\Delta(c_{2})/n}{G(\theta_{2})}\right) \\
        & = \frac{1 - (1-\frac{c_1}{n})^{\rho\cdot n}}{c_1/n}\cdot \frac{c_1}{n}\cdot \left(\theta_{1} +\frac{\Delta(c_{1})/n}{c_1/n}\right) \\
        & \qquad + (1-\frac{c_1}{n})^{\rho\cdot n}\cdot \frac{1 - (1-\frac{c_2}{n})^{(1-\rho)\cdot n}}{c_2/n} \cdot \frac{c_2}{n}\cdot \left(\theta_{2} +\frac{\Delta(c_{2})/n}{c_2/n}\right) \\
        & \geq (1-\epsilon)\cdot \left( 
       (1-e^{-\rho\cdot c_{1}}) \cdot \left(\theta_{1} +\frac{\Delta(c_{1})}{c_{1}}\right) 
       + e^{- \rho \cdot c_1} \cdot (1-e^{-(1-\rho) \cdot c_{2}}) \cdot \left(\theta_{2} +\frac{\Delta(c_{2})}{c_{2}}\right)
       \right),
    \end{align*}
    where in the last inequality we use $1+x\leq e^x$ for all $x>0$ and $1-\frac{c_1}{n} \geq (1-\epsilon)\cdot e^{-\frac{c_1}{n}}$, which holds for sufficiently large $n$, e.g., $n = \omega(\frac{1}{\epsilon})$.
\end{proof}

\subsection{Upper Bounding $\OPT$}\label{ssec:opt-ub-2-thresholds}

Following the argument similar to the one given in Equation~\eqref{eq:opt-upper}, we  obtain the following upper bounds for $\OPT$:
\begin{align*}
    \OPT \leq \theta_{1} + \Delta(c_1) \quad \text{and} \quad
    \OPT \le \theta_{2} + \Delta(c_2).
\end{align*}

However, the above bounds (together with the lower bound on $\ALG$ given in Equation~\eqref{eq:alg-lb-2-thresholds}) alone would not be sufficient to beat $1-1/e$.
Intuitively speaking, this is because we do not have many constraints on the variables $\theta_1, \theta_2, \Delta(c_1)$ and $\Delta(c_2)$ (except that $\theta_1 \geq \theta_2$).
In what follows, we give a more careful analysis that leads to tighter upper bounds for $\OPT$.
We introduce a new variable $\delta$ defined as follows:  
\begin{equation*}
\delta = \int_{\theta_{2}}^{\theta_1} \Pr[x^* \geq t] dt.
\end{equation*}

Recall that $x^* = \max_{i\in[n]} \{x_i\}$.
We have the following upper bound for $\OPT$:
\begin{align}
    \OPT = \bbE[x^*] &= \int_{0}^{\infty} \Pr[x^* \ge t] dt \nonumber \\
    &= \int_{0}^{\theta_2}\Pr[x^* \ge t] dt + \int_{\theta_{2}}^{\theta_{1}} \Pr[x^* \ge t] dt + \int_{\theta_{1}}^{\infty} \Pr[x^* \ge t] dt \nonumber\\
    & \leq \theta_2 + \delta + \bbE[(x^* - \theta_1)^+]
    \leq \theta_2 + \delta + \Delta(c_1). \label{eq:opt-constr-2-thresholds}
\end{align}  

In the next three lemmas, we establish three upper bounds for $\delta$, which,  when 
combined with Equation~\eqref{eq:opt-constr-2-thresholds}, will provide much better upper bounds on \OPT.
Intuitively, the variable $\delta$ helps establish some connection between $\theta_1$ and $\theta_2$, and  between $\Delta(c_1)$ and $\Delta(c_2)$, which can be utilized to derive more constraints on these variables.
Due to space limit, the proofs of Lemmas \ref{lem:delta-ub-1-2-thresholds} to \ref{lem:delta-ub-3-2-thresholds} can be found in Appendix \ref{app:missing-proofs}. 

\begin{lemma}\label{lem:delta-ub-1-2-thresholds}
For arbitrarily small $\epsilon>0$ and sufficiently large $n$, we have 
\begin{equation}\label{eq:delta-ub-1-2-thresholds}
    \delta \le (1+\epsilon)\cdot (\theta_{1}-\theta_{2})\cdot (1-e^{-c_{2}}).
\end{equation}
\end{lemma}

Next, we observe a second upper bound for $\delta$ achieved by using the union bound. 

\begin{lemma}\label{lem:delta-ub-2-2-thresholds}
	We have	\begin{equation}\label{eq:delta-ub-2-2-thresholds}
    	\delta\le \Delta(c_{2})-\Delta(c_{1}).
	\end{equation}
\end{lemma}

So far we have developed two simple upper bounds for $\delta$.
In fact, with these two sets of upper bounds, we can already beat $1-1/e$ and achieve a competitive ratio of at least $0.65$ by using factor revealing LPs.
However, we observe that  the ratio can be further improved by introducing a third upper bound\footnote{As we will discuss in Section~\ref{sec:open-problem}, there is still room for improvement for this upper bound.} on $\delta$.

\begin{lemma}\label{lem:delta-ub-3-2-thresholds}
	Fix any constant $c\in [c_{1},c_{2}]$, let $\gamma := \frac{c-c_1}{e^{-c} - e^{-c_{2}}}$.
	For an arbitrarily small constant $\epsilon>0$ and sufficiently large $n$, we have 
	\begin{equation}\label{eq:delta-ub-3-2-thresholds}
	\gamma \cdot \delta/(1+\epsilon) \le 
	\Delta(c_{2})-\Delta(c_{1}) - (c_{1} - \gamma(1-e^{-c}))\cdot (\theta_1 - \theta_{2}).
	\end{equation}
\end{lemma}

\subsection{Lower Bounding the Competitive Ratios}\label{ssec:comp-ratio-2-thresholds}

So far, we have established a lower bound for $\ALG$ in Section~\ref{ssec:alg-lb-2-thresholds} and several upper bounds for $\OPT$ in Section~\ref{ssec:opt-ub-2-thresholds}.
More importantly, given parameters $c_1,c_2$ and $\rho$, these bounds are linear in $\theta_{1}, \theta_2, \Delta(c_1), \Delta(c_2)$ and $\delta$. 
We can now construct a minimization LP taking Equation~\eqref{eq:alg-lb-2-thresholds} as the objective and Equations~\eqref{eq:opt-constr-2-thresholds}, \eqref{eq:delta-ub-1-2-thresholds}, \eqref{eq:delta-ub-2-2-thresholds}, \eqref{eq:delta-ub-3-2-thresholds} as constraints with $\theta_{1}, \theta_2, \Delta(c_1), \Delta(c_2)$ and $\delta$ being the non-negative LP variables. 
In the following theorem, we prove that the optimal objective of the LP gives a lower bound on the competitive ratio for the algorithm.

\begin{theorem}\label{thm:LP-lb-2-thresholds}
    Given parameters $c_1,c_2$ and $\rho$ of the two-threshold algorithm, the optimal value of the following LP provides a lower bound for the competitive ratio of the algorithm when $n\rightarrow \infty$.
    \begin{align}
		\text{minimize} \quad & \textstyle (1-e^{-\rho\cdot c_{1}}) \cdot \left(\theta_{1} +\frac{\Delta(c_{1})}{c_{1}}\right) 
        + e^{- \rho \cdot c_1} \cdot (1-e^{-(1-\rho) \cdot c_{2}}) \cdot \left(\theta_{2} +\frac{\Delta(c_{2})}{c_{2}}\right) \nonumber\\
		\text{subject to} \quad 
		& \theta_1 \geq \theta_2 \geq 0, \quad
		\Delta(c_1)\geq 0, \quad \delta \geq 0,\nonumber\\
		&  \textstyle  1 \leq \theta_2 + \delta + \Delta(c_1), \nonumber\\
		& \delta \le (\theta_{1}-\theta_{2})\cdot(1-e^{-c_{2}}), \nonumber\\
		& \delta \leq \Delta(c_{2})-\Delta(c_{1}),  \nonumber\\
		& \gamma \cdot \delta \leq \Delta(c_{2})-\Delta(c_{1}) - (c_{1} - \gamma(1-e^{-c}))\cdot (\theta_1 - \theta_{2}), \>\forall c \in [c_{1},c_{2}],  \> \textstyle \gamma = \frac{c-c_{1}}{e^{-c}-e^{-c_{2}}}. \nonumber
	\end{align}
\end{theorem}
\begin{proof}
    The objective of the LP comes from Equation~\eqref{eq:alg-lb-2-thresholds}, the lower bound for $\ALG$, where we omit the $(1-\epsilon)$ term since we have $\epsilon \rightarrow 0$ when $n\rightarrow \infty$.
    We will also omit the $(1+\epsilon)$ terms in the following for the same reason.
    The first set of constraints follows straightforwardly from the definitions of the parameters.
    By scaling we can assume w.l.o.g. that $\OPT=1$. Note that our algorithm does not need to know the value of $\OPT$ to decide the parameters.
    Therefore, Equation~\eqref{eq:opt-constr-2-thresholds} gives the second constraint.
    The three proceeding sets of constraints follow from Equations~\eqref{eq:delta-ub-1-2-thresholds}, \eqref{eq:delta-ub-2-2-thresholds} and \eqref{eq:delta-ub-3-2-thresholds}, the three upper bounds on $\delta$ that we have established in Section~\ref{ssec:opt-ub-2-thresholds}.
    
    Observe that every distribution $D$ induces a set of variables $\{\theta_{1}, \theta_2, \Delta(c_{1}), \Delta(c_{2}), \delta\}$, which will form a feasible solution to the LP as they must abide by the corresponding constraints.
    Since the objective of any feasible solution induced by distribution $D$ provides a lower bound on $\ALG/\OPT$, 
    the optimal (minimum) objective of the LP provides a lower bound on the competitive ratio, i.e., the worst-case performance against all distributions. 
\end{proof}

With the above result, it remains to set the appropriate parameters of the two-threshold algorithm such that the optimal value of the LP is as large as possible.

\begin{theorem}\label{thm:LP-params-2-thresholds}
The two-threshold algorithm achieves a competitive ratio of at least $0.6786$, with parameters $c_1 = 0.7067$, $c_2 = 1.8353$, and $\rho = 0.6204$.
\end{theorem} 
\begin{proof}
    By Theorem~\ref{thm:LP-lb-2-thresholds}, given the values of $c_1$, $c_2$, and $\rho$, the objective value of the factor revealing LP bounds the competitive ratio from below.
    Using an online LP solver\footnote{Online LP solver: \url{https://online-optimizer.appspot.com/}.}, it can be verified that the optimal value of the LP is at least $0.6786$.
\end{proof}

\subsection{An Upper Bound for Two-Threshold Algorithms}\label{ssec:ub-k=2}

Finally, we complement our algorithmic results with an upper bound of $0.7081$ on the competitive ratio for any algorithm that uses two thresholds from two queries.
Specifically, we construct three distributions and prove that any such algorithm, parameterized by $c_1,c_2$ and $\rho$, cannot achieve a competitive ratio better than $0.7081$ in all three distributions. 

\begin{lemma} \label{lem:2-threshold:upper}
Any algorithm that defines two thresholds by making two queries cannot achieve a competitive ratio better than $0.7081$.
\end{lemma}

\section{Beating $1-1/e$ using a Single Query}
\label{sec:observe-and-accept}

	In this section, we show that it is possible to beat the competitive ratio of $1-1/e$ using just one query. 
	Our algorithm combines the single-threshold algorithm and the observe-then-accept algorithm for the secretary problem.
	Let $c\in [0,n]$ and $\rho\in [0,1]$ be the parameters of the algorithm.
	Again, for sufficiently large $n$ we can assume w.l.o.g. that $\rho \cdot n$ is an integer.
	The algorithm, denoted by {\em observe-and-accept}, works as follows.
	Let $\theta_1 = v(1-c/n)$. 
	We have $\Pr_{x\sim D}[x\ge \theta_1] = G(\theta_1) = c/n$.
	In the quantile-based phase, for $i=1,2,\ldots,\rho\cdot n$, we accept $x_i$ if and only if $x_i\geq \theta_1$.
	If no variables are accepted in the quantile-based phase, we enter the observation-based phase and let $\theta_2 = \max_{i\leq \rho\cdot n}\{x_i\}$. Note that $\theta_2$ is a random variable and $\theta_2 < \theta_1$.
 	For $i=\rho\cdot n+ 1,\ldots,n$, we accept $x_i$ if and only if $x_i\geq \theta_2$ (see Algorithm \ref{alg:Observe-and-Accept}).

\begin{algorithm}[ht]
\caption{Observe-and-Accept}\label{alg:Observe-and-Accept}
\textbf{Input:} {Number of items $n$, quantile choice $c$ and observation portion $\rho$.}

\textbf{Output:} {A variable selected by the gambler.}
	
Set $\theta_1 := v(1-c/n)$ by making a query to the oracle; 
	\For{$i = 1,2,\ldots,\rho\cdot n$}{ 
			   \If{$x_i \geq \theta_1$}{
		   \textbf{return} $x_i$ and \textbf{terminate}; \qquad\qquad {\color{gray}// quantile-based phase}}}

Set $\theta_2 := \max_{i\leq \rho\cdot n}\{x_i\}$;
	
	\For{$i = \rho\cdot n + 1,\ldots,n$}
 {	\If{$x_i \geq \theta_2$}{ 
		  \textbf{return} $x_i$ and \textbf{terminate}; \qquad\qquad {\color{gray}// observation-based phase}}
}
\end{algorithm}
	
	\begin{theorem}\label{thm:single-qeury-beat-1-1/e}
		There exists a $0.6718$-competitive algorithm for the prophet inequality problem on unknown i.i.d.  distributions that makes a single query to the distribution oracle.
	\end{theorem}
	
	We prove the theorem in the following two subsections.
	
    \subsection{Bounding $\ALG$ and $\OPT$}
	
    We first give lower bounds for $\ALG$ and upper bounds for $\OPT$, which will then become the constraints of the factor revealing LP we study in the next subsection.
    Recall Equation~\eqref{eq:alg-decomp} from Section~\ref{sec:warmup}, in which we have $\ALG = \sum_{i=1}^n \Pr[\tau = i]\cdot \bbE[x_i \mid \tau = i]$,
	where $\tau \in [n+1]$ is the stopping time of the algorithm.
	For $i\leq \rho\cdot n$, the algorithm accepts  variable $x_i$ if its realization is at least $\theta_1$, which happens with probability $c/n$.
	By a similar analysis seen in Equations~\eqref{eq:alg-prob} and \eqref{eq:alg-gain}, we have 
	\begin{align*}
	    \Pr[\tau = i] = \left(1-\frac{c}{n}\right)^{i-1}\cdot \frac{c}{n}, & \quad \forall i \le \rho\cdot n, \quad \text{ and } \\
		\bbE[x_i \mid \tau = i]=  \theta_1 +  \frac{\Delta(c)}{c}, & \quad \forall i \le \rho\cdot n.
	\end{align*}
	Therefore, the expected gain of the algorithm from the quantile-based phase is given by
	\begin{align}
		\Lambda_1 & := \sum_{i=1}^{\rho n} \Pr[\tau = i]\cdot \bbE[x_i \mid \tau = i] = \sum_{i=1}^{\rho n} (1-\frac{c}{n})^{i-1}\cdot \left( \frac{c}{n}\cdot \theta_1 + \frac{\Delta(c)}{n}  \right) \nonumber \\
		& = \frac{1-(1-c/n)^{\rho n}}{c/n} \cdot \left(\frac{c}{n} \cdot \theta_1 + \frac{\Delta(c)}{n} \right) = \left( 1-(1-\frac{c}{n})^{\rho n} \right) \cdot \left(\theta_1 + \frac{\Delta(c)}{c}  \right) \nonumber \\
		& \geq  \left( 1 - e^{-c \rho} \right) \cdot \left(\theta_1 + \frac{\Delta(c)}{c}  \right). \label{equation:lower-bound-Lambda-1}
	\end{align}
	
	Next, we give a lower bound for the expected gain of the algorithm from the observation-based phase, which is given by $\Lambda_2 := \sum_{i=\rho n+1}^n \Pr[\tau = i]\cdot \bbE[x_i \mid \tau = i]$.
	Unlike the quantile-based phase, the threshold $\theta_2 = \max_{i\leq \rho n} \{ x_i \}$ is a random variable and thus $F(\theta_2)$ is not a fixed probability.
	Our algorithm enters the observation-based phase if and only if $\theta_2 < \theta_1$, i.e., no variable passes the threshold $\theta_1$ in the quantile-based phase.
	Note that
	\begin{equation*}
		\Pr[\theta_2 < \theta_1] = \Pr[\forall i\leq \rho n, x_i < \theta_1] = (1-\frac{c}{n})^{\rho n} \geq (1-\epsilon)\cdot e^{-c \rho},
		\end{equation*}
	where $\epsilon>0$ is an arbitrarily small constant given that $n$ is sufficiently large.
	Moreover, given a realization of $\theta_2 < \theta_1$, we can express the gain of the algorithm $h(\theta_2)$ in the observation-based phase using a similar argument as above.
	Specifically, conditioned on a given $\theta_2 < \theta_1$, we have
	\begin{equation*}
		h(\theta_2) := \sum_{i=\rho n+1}^{n} \Pr[\tau = i \mid \tau> \rho n]\cdot \bbE[x_i \mid \tau = i]  
		= \sum_{i=1}^{(1-\rho)n} F(\theta_2)^{i-1}\cdot \left(G(\theta_2)\cdot \theta_2 + \bbE[(x - \theta_2)^+]\right).
	\end{equation*}
	In summary, we can lower bound the expected total gain of the algorithm by
	\begin{align*}
		\ALG & = \left( 1-(1-\frac{c}{n})^{\rho n} \right) \cdot \left(\theta_1 + \frac{\Delta(c)}{c}  \right) + \left(1-\frac{c}{n}\right)^{\rho n}\cdot \bbE_{\theta_2} [h(\theta_2) \mid \theta_2 < \theta_1] \\
		& \geq \left( 1-e^{-c \rho} \right) \cdot \left(\theta_1 + \frac{\Delta(c)}{c}\right) + (1-\epsilon)\cdot e^{- c \rho } \cdot \bbE_{\theta_2} [h(\theta_2) \mid \theta_2 < \theta_1].
	\end{align*}
	
	Unfortunately, it is challenging to derive a closed-form expression of the conditional expectation $\bbE_{\theta_2} [h(\theta_2) \mid \theta_2 < \theta_1]$ in terms of $c$ and $\rho$.
	To overcome this difficulty, we relax this term and lower bound it 
	using a factor revealing LP.
	Specifically, consider the case when
	\begin{equation*}
		v(1-\frac{\beta_2}{n}) \leq \theta_2 \leq v(1-\frac{\beta_1}{n}),
	\end{equation*}
	where $\beta_1,\beta_2$ are constants satisfying $c < \beta_1 < \beta_2$.
	Recall that $F(v)$ and $v(q) = F^{-1}(q)$ are increasing functions of $v$ and $q$, and $G(v)$ and $\bbE[(x-\theta)^+]$ are decreasing functions of $v$ and $\theta$.
	We can lower bound $h(\theta_2)$ by
	\begin{align*}
		h(\theta_2) & = \sum_{i=1}^{(1-\rho)n} F(\theta_2)^{i-1}\cdot \left(G(\theta_2)\cdot \theta_2 + \bbE[(x - \theta_2)^+]\right) \\
		& \geq \sum_{i=1}^{(1-\rho)n} F\left( v(1-\frac{\beta_2}{n}) \right)^{i-1}\cdot \left\{G\left( v(1-\frac{\beta_1}{n}) \right)\cdot v(1-\frac{\beta_2}{n}) + \bbE\left[ \left(x - v(1-\frac{\beta_1}{n}) \right)^+ \right]\right\} \\
		& \geq \sum_{i=1}^{(1-\rho)n} \left( 1-\frac{\beta_2}{n} \right)^{i-1}\cdot \left( \frac{\beta_1}{n} \cdot v(1-\frac{\beta_2}{n}) + \frac{\Delta(\beta_1)}{n}  \right) \\
		& \geq \frac{1- \left( 1-\frac{\beta_2}{n} \right)^{(1-\rho)n}}{\beta_2 / n} \cdot \left( \frac{\beta_1}{n} \cdot v(1-\frac{\beta_2}{n}) + \frac{\Delta(\beta_1)}{n} \right) \\
		& \geq \left( 1- \left( 1-\frac{\beta_2}{n} \right)^{(1-\rho)n} \right) \cdot \left( \frac{\beta_1}{\beta_2}\cdot v(1-\frac{\beta_2}{n}) + \frac{ \Delta(\beta_1)}{\beta_2}  \right) \\
		& \geq \left( 1- e^{-\beta_2 (1-\rho)} \right) \cdot \left( \frac{\beta_1}{\beta_2}\cdot v(1-\frac{\beta_2}{n}) + \frac{\Delta(\beta_1)}{\beta_2}  \right) 
		:= H(\beta_1, \beta_2).
	\end{align*}
	
	The important observation here is that $H(\beta_1,\beta_2)$ depends only on $c, \beta_1$, $\beta_2, n$ and the distribution $D$, and is independent of $\theta_2$.
	Moreover, it is linear in $v(1-\beta_2/n)$ and $\Delta(\beta_1)$.
	Note that conditioned on $\theta_2 < \theta_1$, the event of $v(1-\beta_2/n) \leq \theta_2 \leq v(1-\beta_1/n)$, where $c<\beta_1<\beta_2$, happens with probability
	\begin{align*}
		&\ \Pr\left[v(1-\frac{\beta_2}{n}) \leq \theta_2 \leq v(1-\frac{\beta_1}{n}) \bigm| \theta_2 < \theta_1\right] \\
		= &\ \frac{\Pr[v(1-\frac{\beta_2}{n}) \leq \theta_2 \leq v(1-\frac{\beta_1}{n})]}{\Pr[\theta_2 < \theta_1]}  =  \frac{(1-\frac{\beta_1}{n})^{\rho n} - (1-\frac{\beta_2}{n})^{\rho n}}{(1-\frac{c}{n})^{\rho n}} := p(\beta_1, \beta_2).
	\end{align*}
		
	Therefore, for any sequence $c = \beta_0 < \beta_1 < \beta_2 < \cdots < \beta_k$ and sufficiently large $n$,  we can lower bound the gain of  the observation-based phase by
	\begin{align}
		\Lambda_2 & \geq \left(1-\frac{c}{n}\right)^{\rho n}\cdot \sum_{i=1}^k p(\beta_{i-1},\beta_i)\cdot H(\beta_{i-1},\beta_i) 
\nonumber \\
    & = \sum_{i=1}^k \left( (1-\frac{\beta_{i-1}}{n})^{\rho n} - (1-\frac{\beta_i}{n})^{\rho n} \right) \cdot H(\beta_{i-1},\beta_i) \nonumber \\ 
		& \geq (1-\epsilon) \sum_{i=1}^k \left( e^{-\beta_{i-1} \rho} - e^{-\beta_i \rho} \right) \cdot \left( 1- e^{-\beta_i (1-\rho)} \right) \cdot \left( \frac{\beta_{i-1}}{\beta_i}\cdot v(1-\frac{\beta_i}{n}) + \frac{\Delta(\beta_{i-1})}{\beta_i}  \right). \label{equation:lower-bound-Lambda-2}
	\end{align}
	
	So far, we have given a lower bound for $\ALG$ in Equations~\eqref{equation:lower-bound-Lambda-1} and~\eqref{equation:lower-bound-Lambda-2}.
	In the following, we give upper bounds for $\OPT$.
	Following similar analyses as in Section~\ref{ssec:2-queries}, we define
	\begin{equation*}
	    \delta_i = \int_{v(1-\frac{\beta_i+1}{n})}^{v(1-\frac{\beta_i}{n})} \Pr[x\geq t] dt,\qquad  \forall 0\leq i \leq k-1.
	\end{equation*}
	Then we have
	\begin{align*}
	    \OPT & = \int_{0}^{v(1-\frac{\beta_k}{n})} \Pr[x^* \geq t] dt + \sum_{i=0}^{k-1} \int_{v(1-\frac{\beta_{i+1}}{n})}^{v(1-\frac{\beta_i}{n})} \Pr[x^* \geq t] dt + \int_{v(1-\frac{\beta_0}{n})}^{\infty} \Pr[x^* \geq t] dt \\ 
	    &\leq v(1-\frac{\beta_k}{n}) + \sum_{i=0}^{k-1} \delta_i + \Delta(\beta_0).
	\end{align*}
	
	Following the same analyses for proving Lemmas~\ref{lem:delta-ub-1-2-thresholds},~\ref{lem:delta-ub-2-2-thresholds} and~\ref{lem:delta-ub-3-2-thresholds}, we obtain the following: for arbitrarily small $\epsilon>0$ and sufficiently large $n$ and for all $i\in\{ 0,1,\ldots,k-1 \}$,
	\begin{align*}
	    \delta_i & \leq (1+\epsilon)\cdot(1-e^{-\beta_{i+1}})\cdot \left(v(1-\frac{\beta_i}{n}) - v(1-\frac{\beta_i+1}{n}) \right), \\
	    \delta_i & \leq \Delta(\beta_{i+1}) - \Delta(\beta_i).
	\end{align*}
	For all $\zeta\in [\beta_i, \beta_{i+1}]$ and $\gamma = \frac{\zeta - \beta_{i}}{e^{-\zeta} - e^{-\beta_{i+1}}}$, we have
	\begin{equation*}
	    \gamma\cdot \delta_i/(1+\epsilon) \leq \Delta(\beta_{i+1}) - \Delta(\beta_i) - (\beta_i - \gamma(1-e^{-\zeta}))\cdot \left(v(1-\frac{\beta_i}{n}) - v(1-\frac{\beta_i+1}{n}) \right).
	\end{equation*}

	\subsection{Factor Revealing LP}
	\label{sec:single:FactorRevealingLP}
	We finish the proof of Theorem~\ref{thm:single-qeury-beat-1-1/e} by formulating the lower and upper bounds into a factor revealing LP and choosing appropriate parameters to obtain a lower bound of $0.6718$ on the competitive ratio.
	In the following, we fix $\beta_i = \beta^i\cdot c$ for all $i=0,1,\ldots, k$, for some $\beta>1$.
	For ease of notation, 
	we use $v_i = v(1-\beta_i/n)$ and $\Delta_i = \Delta(\beta_i)$, for all $i=0,1,\ldots,k$.
	Note that $c= \beta_0$, $\theta_1 = v(1-\beta_0/n) = v_0$ and $\Delta(c) = \Delta_0$.
	Recall that $v_0 \geq v_1\geq \cdots \geq v_k$ and $\Delta_0 \leq \Delta_1\leq \cdots \leq \Delta_k$.
	By fixing $c$ and $k$, we uniquely define a sequence $\beta_0, \beta_1,\ldots,\beta_k$. However, the values of $\{ v_i,\Delta_i \}_{0\leq i\leq k}$ depend on the unknown  distribution $D$. 
	As argued in Section~\ref{ssec:comp-ratio-2-thresholds},
	by taking the lower bound on $\ALG$ as the objective and upper bounds on $\OPT$ as  constraints, we obtain an LP whose optimal value provides a lower bound for the competitive ratio of the algorithm.
	The LP variables are $\{ v_i,\Delta_i \}_{0\leq i\leq k}$ and $\{ \delta_i \}_{0\leq i< k}$.

	\smallskip
	
	Recall that we have lower bounded $\ALG = \Lambda_1 + \Lambda_2$, where 
	\begin{align*}
		\Lambda_1 & \geq  \left( 1 - e^{-\beta_0 \rho} \right) \cdot (v_0 + \Delta_0/\beta_0 ),\qquad \\
		\Lambda_2 & \geq (1-\epsilon) \sum_{i=1}^k \left( e^{-\beta_{i-1} \rho} - e^{-\beta_i \rho} \right) \cdot \left( 1- e^{-\beta_i (1-\rho)} \right) \cdot ( (\beta_{i-1}/\beta_i)\cdot v_i +  \Delta_{i-1}/\beta_i ),
	\end{align*}
 are both linear in $\{ v_i,\Delta_i \}_{0\leq i\leq k}$.
	We have also upper bounded $\OPT$ by 
	\begin{equation*}
	    \OPT \leq v(1-\frac{\beta_k}{n}) + \sum_{i=0}^{k-1} \delta_i + \Delta_0,
	\end{equation*}
	and each $\delta_i$, where $i\in\{0,1,\ldots,k-1\}$, by (recall that $\gamma = \frac{\zeta - \beta_{i}}{e^{-\zeta} - e^{-\beta_{i+1}}}$)
	\begin{align*}
	    \delta_i & \leq (1+\epsilon)\cdot(1-e^{-\beta_{i+1}})\cdot (v_i - v_{i+1}), \\
	    \delta_i & \leq \Delta(\beta_{i+1}) - \Delta(\beta_i), \\
	    \gamma\cdot \delta_i/(1+\epsilon) & \leq \Delta_{i+1} - \Delta_i - (\beta_i - \gamma(1-e^{-\zeta}))\cdot (v_i - v_{i+1}), \quad \forall \zeta\in [\beta_i,\beta_{i+1}].
	\end{align*}
	
	Again, we can assume w.l.o.g. that $\OPT = 1$.
	Since $\epsilon>0$ can be arbitrarily small when $n\rightarrow \infty$, the terms $(1-\epsilon)$ and $(1+\epsilon)$ can be removed.
	The factor revealing LP is as follows.
	\begin{align*}
		\text{minimize} \quad & \Lambda_1 + \Lambda_2 \\
		\text{subject to} \quad 
		& \textstyle  \Lambda_1 \geq  ( 1 - e^{-\beta_0 \rho} ) \cdot (v_0 + \Delta_0/\beta_0 ), \\
		&  \textstyle \Lambda_2 \geq  \sum_{i=1}^k \left( e^{-\beta_{i-1} \rho} - e^{-\beta_i \rho} \right) \cdot \left( 1- e^{-\beta_i (1-\rho)} \right) \cdot ( (\beta_{i-1}/\beta_i)\cdot v_i +  \Delta_{i-1}/\beta_i ), \\
		& \textstyle  1 \geq v_k + \sum_{i=0}^{k-1}\delta_i + \Delta_0 \\
		& \textstyle  \delta_i \leq \left(1- e^{-\beta_{i+1}} \right) \cdot ( v_{i} - v_{i+1} ), \qquad \forall 0\leq i\leq k-1  \\
		& \delta_i \leq \Delta_{i+1} - \Delta_i, \qquad \forall 0\leq i\leq k-1 \\
		& \gamma\cdot \delta_i \leq \Delta_{i+1} - \Delta_i - (\beta_i - \gamma(1-e^{-\zeta}))\cdot (v_i - v_{i+1}), \quad \forall 0\leq i\leq k-1, \zeta\in[\beta_i,\beta_{i+1}] \\
		& v_0 \geq v_1 \geq \cdots \geq v_k \geq 0,\\
		& \Delta_0,\delta_0,\ldots,\delta_{k-1} \geq 0,
	\end{align*}
	where $\rho$, $\{\beta_i\}_{0\leq i\leq k}$ and $\gamma = \frac{\zeta - \beta_{i}}{e^{-\zeta} - e^{-\beta_{i+1}}}$ are constants; $\Lambda_1, \Lambda_2$, $\{v_i, \Delta_i\}_{0\leq i\leq k}$, $\{\delta_i\}_{0\leq i< k}$ are variables of the LP.
	As argued, every distribution $D$ induces a feasible solution to the LP.
	Therefore, the optimal value of the LP provides a lower bound for the competitive ratio of our algorithm when $n\rightarrow \infty$.
	
	The following claim is verified by the online LP solver, which completes the proof of Theorem~\ref{thm:single-qeury-beat-1-1/e}.
	
	\begin{claim}\label{claim:single:opt}
		By fixing constants $c = 0.72941, \rho = 0.64863$, $k=100$ and $\beta_i = 1.03^i\cdot c$ for all $0\leq i\leq k$, the optimal solution to the above LP has objective at least $0.6718$.
	\end{claim}
	
    \paragraph{Remark}
	Suppose we are given a fixed quantile (i.e., a fixed  $c$) instead of being able to choose one, we may still beat $1-1/e$.
	Let us revisit the example in Figure~\ref{fig:c_ratio} in the following updated Figure~\ref{fig:c_ratio_rho}.
	By setting $k=100$ and $\beta_i = 1.03^i\cdot c$ for all $0\leq i\leq k$, we compute the corresponding separation point $\rho$ (the green-cross curve) and the optimal reward using the above LP (the blue-dot curve). 
	We see that for any $c$ between $0.5$ and $1$, the competitive ratio is better than $1-1/e$.
	This suggests that if we are not allowed to make queries but are given a fixed quantile, as long as it is within this range, the observe-and-accept algorithm can still do strictly better than the optimal single-threshold algorithm. 
	Taking a closer look at the cases when $c>1$, we see that the optimal $\rho$ is always 1, which means the algorithm only has the quantile-based phase using $v(1-c/n)$ as the threshold.
	This is because for $c>1$, $v(1-c/n)$ is already smaller than the threshold used in the optimal single-threshold algorithm, and thus it is worse off to add the observation-based phase whose threshold is even smaller than $v(1-c/n)$.
	{Finally, it is clear to note that if the given $c$ is very large, which means that the threshold for the quantile-based phase is  too small, the observe-and-accept algorithm may accept a very small value and perform badly.
	As shown in Figure~\ref{fig:c_ratio_rho}, when $c$ is greater than $2.5$, the competitive ratio for observe-and-accept is less than $1/e$.}
 
	\begin{figure}[ht]
	    \centering
	    \includegraphics[width=0.7\textwidth]{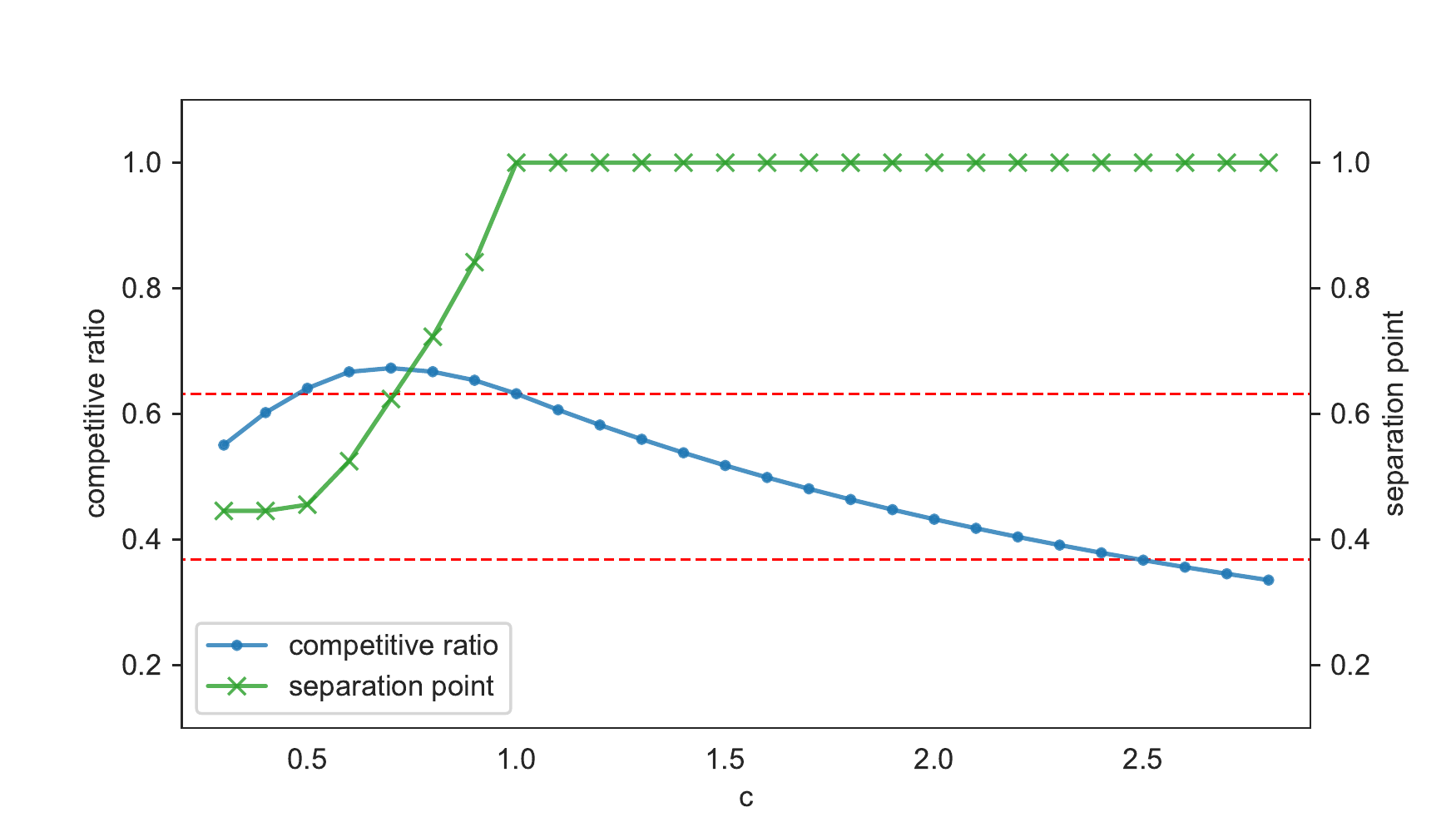}
	    \caption{Competitive ratios, as well as the corresponding best separation point $\rho$, for different quantiles at $1 - c/n$, by setting $k=100$ and $\beta_i = 1.03^i\cdot c$ for all $0\leq i\leq k$. (Figure \ref{fig:c_ratio} revisited).}
	    \label{fig:c_ratio_rho}
	\end{figure}

	\subsection{Upper Bound on the Competitive Ratio of the Algorithm}
	
	To complement our lower bound on the competitive ratio of the algorithm, we also explore upper bounds for this type of algorithm.
	In particular, we show the following.
	
	\begin{lemma}\label{lemma:upper-bound-observe-and-choose}
		No observe-and-accept algorithm can do better than $0.6921$-competitive.
	\end{lemma}
	\begin{proof}
		Consider any observe-and-accept algorithm with parameters $c$ and $\rho$.
		In the following, we construct two distributions $D_1$ and $D_2$, and show that the performance of the algorithm on one of the two distributions is at most $0.6921$.
		Note that by  construction both distributions have strictly increasing CDF and satisfy $\OPT = 1$.
		
		Let the first distribution $D_1$ be a uniform distribution over $[1-\epsilon,1+\epsilon]$, for some arbitrarily small $\epsilon >0$.
		Then we have
		\begin{align}
			\ALG & \leq 1+\epsilon - \Pr[\tau > n] = 1+\epsilon - \Pr[\forall i\leq \rho n, x_i<\theta_1 \text{ and } \forall i > \rho n, x_i < \max_{j\leq \rho n} \{x_j\}] \nonumber\\
			& = 1 + \epsilon - (1-\frac{c}{n})^n \cdot \rho \leq 1 + \epsilon - \rho\cdot e^{-c}, \label{equation:upper-bound-D1}
		\end{align}
		where the last equality holds because $\tau > n$ happens if and only if all variables are below $\theta_1$ and the maximum of them appears in the first $\rho\cdot n$ variables.
		
		Let the second distribution $D_2$ be as follows.
		Let $M \gg n$ be an arbitrarily large number and $\epsilon > 0$ be arbitrarily small.
		With probability $\frac{1}{n\cdot M}$,  $x$ is uniformly distributed over $[M-\epsilon, M+\epsilon]$; with probability $1-\frac{1}{n\cdot M}$,  $x$ is uniformly distributed over $[0, \epsilon]$.
		Note that when $M\rightarrow \infty$ and $\epsilon \rightarrow 0$,
		\begin{equation*}
		\OPT \approx M\cdot \left( 1 - (1-\frac{1}{n\cdot M})^n \right) = M\cdot \frac{1}{M} = 1. 
		\end{equation*}
		Observe that for distribution $D_2$, we have $\theta_1 \leq \epsilon$ and $\Delta(c) = n\cdot \bbE[(x-\theta_1)^+] \approx 1$.
		Since $\theta_1$ is close to $0$, the gain of the algorithm is determined by how likely the algorithm accepts a variable with a value close to $M$.
		Therefore, we can upper bound $\ALG$ by
		\begin{equation*}
			\ALG < 
			\epsilon + M\cdot\sum_{i=1}^n \Pr[\tau > i-1]\cdot \Pr[x_i \geq M-\epsilon].
		\end{equation*}
		Note that $\Pr[x_i \geq M-\epsilon] = \frac{1}{n\cdot M}$ for all $i$. 
		For $i\leq \rho n$, we have
		\begin{equation*}
			\Pr[\tau > i-1] = (1-\frac{c}{n})^{i-1}.
		\end{equation*}
		For $i > \rho n$, the event $\tau > i-1$ happens if and only if all variables $x_1,\ldots,x_{i-1}$ are below $\theta_1$, and the maximum of them appears before index $\rho n$. Hence, 
		\begin{equation*}
			\Pr[\tau > i-1] = (1-\frac{c}{n})^{i-1}\cdot \frac{\rho n}{i-1}.
		\end{equation*}
		Putting everything together, as $n\rightarrow \infty$, we get
		\begin{align}
			\ALG & < \epsilon + \frac{1}{n}\cdot\sum_{i=1}^n \Pr[\tau > i-1] 
             = \epsilon + \frac{1}{n}\cdot\sum_{i=1}^{\rho n} (1-\frac{c}{n})^{i-1}
			+ \frac{1}{n}\cdot\sum_{i=\rho n + 1}^{n} (1-\frac{c}{n})^{i-1}\cdot \frac{\rho n}{i-1} \nonumber \\
			& = \epsilon + \frac{1 - (1-\frac{c}{n})^{\rho n}}{c} + \sum_{i=\rho n + 1}^{n} (1-\frac{c}{n})^{i-1}\cdot \frac{\rho}{i-1} 
            \approx \epsilon +\frac{1-e^{-c \rho}}{c} + \rho\cdot \int_\rho^1 \frac{e^{-c t}}{t} dt. \label{equation:upper-bound-D2} 
		\end{align}
		
		It is not difficult to check that when $c>1$ the above upper bound is less than $1-1/e$. Thus any observe-and-accept algorithm with a competitive ratio $ > 1-1/e$ must have $c<1$.
		It remains to show that when $\epsilon \rightarrow 0$, $n\rightarrow \infty$ and $M\rightarrow\infty$, for any $c,\rho\in[0,1]$, at least one of the two upper bounds given in Equations~\eqref{equation:upper-bound-D1} and~\eqref{equation:upper-bound-D2} has a value less than $0.6921$:
		\begin{equation}
			\max_{c,\rho\in[0,1]} \left\{ \min \left\{ 1 - \rho\cdot e^{-c}, \frac{1-e^{-c \rho}}{c} + \rho\cdot \int_\rho^1 \frac{e^{-c t}}{t} dt \right\} \right\} \leq 0.6921. \label{equation:upper-bound-c-rho}
		\end{equation}
		
		Using a similar argument as in Section~\ref{ssec:ub-k=2}, we  prove Equation~\eqref{equation:upper-bound-c-rho} by discretizing the domain for $c$ and $\rho$ and using computational tools.
		For example, letting $I = \{ \frac{i}{10^5} : 1\leq i\leq 10^5 \}$, we can check that the LHS of Equation~\eqref{equation:upper-bound-c-rho} can be approximated by replacing ``$c,\rho\in [0,1]$'' with ``$c,\rho\in I$'', with an additive error of at most $3\times 10^{-5}$.
		By enumerating all $c,\rho\in I$ and recording the maximum, we get an upper bound of $0.69204$, which is achieved when $c = 0.37476$ and $\rho = 0.44799$.
		Combining the above discretized upper bound with the additive error of $3\times 10^{-5}$ gives an upper bound of $0.69207 < 0.6921$. 
	\end{proof}

\section{Conclusion and Open Problems}\label{sec:open-problem}


In this work, we study the single-choice prophet inequality problem with unknown i.i.d. distributions.
We propose a new model in which an algorithm has access to an oracle that answers quantile queries about the distribution, with which we implement the multi-threshold blind strategies and show that with $k\geq 2$ thresholds, the competitive ratio is at least $0.6786$.
Moreover, we demonstrate that with queries, one can do more than just implementations of blind strategies, by proposing an algorithm that uses a single query to achieve a competitive ratio of $0.6718$.

Our work uncovers many interesting open problems and future research directions. 

The first direction is to improve the competitive ratios and the upper bounds.
We believe that our analysis for the competitive ratio does not fully release the potential of the proposed algorithms, and improvements on the lower bounds of the competitive ratios are possible.
In particular, our current technique of factor revealing LP hinders us from deriving tighter bounds on $\ALG$ and $\OPT$.
For example, the upper bound on $\delta$ we present in Lemma~\ref{lem:delta-ub-3-2-thresholds} is only a linear approximation of a stronger and more general (non-linear) bound. 
In order to formulate this upper bound into a linear constraint of the LP, we have to relax it into a linear form.
It would be very interesting to investigate how much the competitive ratio can be improved if the analysis is able to incorporate such non-linear constraints.

Another natural open problem is to study more general forms of the observe-and-accept algorithm.
For example, one can generalize the algorithm from a single query to multiple queries.
Based on our current results and analyses, we believe that with two queries, it is possible to derive algorithms with a competitive ratio strictly above $0.68$.
One can also investigate alternative ways to utilize the queries.
For instance, is it beneficial to introduce more phases for the single-query algorithm?  
That is, the first phase uses the threshold derived from a query to the oracle while each later phase uses the maximum realization from the previous phase as the threshold.
Furthermore, if having more phases helps, 
what is the best possible competitive ratio one can achieve with a single query? 

Finally, it is interesting to investigate the prophet secretary problem under the query model.
In the prophet secretary problem, the $n$ variables $x_1,x_2,\ldots,x_n$ are drawn from $n$ (possibly different) distributions $D_1,D_2,\ldots,D_n$, and they arrive following a uniformly-at-random chosen order.
The problem generalizes the prophet inequality problem on i.i.d. distributions.
It can be shown that a single query on the distribution $D^*$ of $x^* = \max_{1\leq i\leq n} \{ x_i \}$ suffices to achieve a competitive ratio of $1-1/e$.
Whether we can beat $1-1/e$ with a single query on $D^*$ would be an interesting open problem to investigate.


{
 	\bibliography{references}

\newcommand{\etalchar}[1]{$^{#1}$}
\begin{thebibliography}{JMM{\etalchar{+}}03}

\bibitem[ABB21]{conf/sigecom/AllouahBB21}
Amine Allouah, Achraf Bahamou, and Omar Besbes.
\newblock Optimal pricing with a single point.
\newblock In {\em {EC}}, page~50. {ACM}, 2021.

\bibitem[ACK18]{sigecom/AzarCK18}
Yossi Azar, Ashish Chiplunkar, and Haim Kaplan.
\newblock Prophet secretary: Surpassing the 1-1/e barrier.
\newblock In {\em {EC}}, pages 303--318. {ACM}, 2018.

\bibitem[AEE{\etalchar{+}}17]{stoc/AbolhassaniEEHK17}
Melika Abolhassani, Soheil Ehsani, Hossein Esfandiari, MohammadTaghi
  Hajiaghayi, Robert~D. Kleinberg, and Brendan Lucier.
\newblock Beating 1-1/e for ordered prophets.
\newblock In {\em {STOC}}, pages 61--71. {ACM}, 2017.

\bibitem[AKW19]{geb/AzarKW19}
Pablo~Daniel Azar, Robert Kleinberg, and S.~Matthew Weinberg.
\newblock Prior independent mechanisms via prophet inequalities with limited
  information.
\newblock {\em Games Econ. Behav.}, 118:511--532, 2019.

\bibitem[BC22]{DBLP:journals/corr/abs-2211-04145}
Archit Bubna and Ashish Chiplunkar.
\newblock Prophet inequality: Order selection beats random order.
\newblock {\em CoRR}, abs/2211.04145, 2022.

\bibitem[CDF{\etalchar{+}}21]{corr/abs-2111-03174}
Constantine Caramanis, Paul D{\"{u}}tting, Matthew Faw, Federico Fusco, Philip
  Lazos, Stefano Leonardi, Orestis Papadigenopoulos, Emmanouil Pountourakis,
  and Rebecca Reiffenh{\"{a}}user.
\newblock Single-sample prophet inequalities via greedy-ordered selection.
\newblock {\em CoRR}, abs/2111.03174, 2021.

\bibitem[CDFS19]{ec/CorreaDFS19}
Jos{\'{e}}~R. Correa, Paul D{\"{u}}tting, Felix~A. Fischer, and Kevin Schewior.
\newblock Prophet inequalities for {I.I.D.} random variables from an unknown
  distribution.
\newblock In {\em {EC}}, pages 3--17. {ACM}, 2019.

\bibitem[CFH{\etalchar{+}}17]{sigecom/CorreaFHOV17}
Jos{\'{e}}~R. Correa, Patricio Foncea, Ruben Hoeksma, Tim Oosterwijk, and Tjark
  Vredeveld.
\newblock Posted price mechanisms for a random stream of customers.
\newblock In {\em {EC}}, pages 169--186. {ACM}, 2017.

\bibitem[CFH{\etalchar{+}}18]{sigecom/CorreaFHOV18}
Jos{\'{e}}~R. Correa, Patricio Foncea, Ruben Hoeksma, Tim Oosterwijk, and Tjark
  Vredeveld.
\newblock Recent developments in prophet inequalities.
\newblock {\em SIGecom Exch.}, 17(1):61--70, 2018.

\bibitem[CHMS10]{stoc/ChawlaHMS10}
Shuchi Chawla, Jason~D. Hartline, David~L. Malec, and Balasubramanian Sivan.
\newblock Multi-parameter mechanism design and sequential posted pricing.
\newblock In {\em {STOC}}, pages 311--320. {ACM}, 2010.

\bibitem[CLLL22]{journals/ai/ChenLLL22}
Jing Chen, Bo~Li, Yingkai Li, and Pinyan Lu.
\newblock Bayesian auctions with efficient queries.
\newblock {\em Artif. Intell.}, 303:103630, 2022.

\bibitem[CSZ19]{soda/CorreaSZ19}
Jos{\'{e}}~R. Correa, Raimundo Saona, and Bruno Ziliotto.
\newblock Prophet secretary through blind strategies.
\newblock In {\em {SODA}}, pages 1946--1961. {SIAM}, 2019.

\bibitem[DFKL20]{siamcomp/DuttingFKL20}
Paul D{\"{u}}tting, Michal Feldman, Thomas Kesselheim, and Brendan Lucier.
\newblock Prophet inequalities made easy: Stochastic optimization by pricing
  nonstochastic inputs.
\newblock {\em {SIAM} J. Comput.}, 49(3):540--582, 2020.

\bibitem[DK15]{esa/DuttingK15}
Paul D{\"{u}}tting and Robert Kleinberg.
\newblock Polymatroid prophet inequalities.
\newblock In {\em {ESA}}, volume 9294 of {\em Lecture Notes in Computer
  Science}, pages 437--449. Springer, 2015.

\bibitem[DK19]{ec/DuttingK19}
Paul D{\"{u}}tting and Thomas Kesselheim.
\newblock Posted pricing and prophet inequalities with inaccurate priors.
\newblock In {\em {EC}}, pages 111--129. {ACM}, 2019.

\bibitem[DKL20]{focs/DuttingKL20}
Paul D{\"{u}}tting, Thomas Kesselheim, and Brendan Lucier.
\newblock An o(log log m) prophet inequality for subadditive combinatorial
  auctions.
\newblock In {\em {FOCS}}, pages 306--317. {IEEE}, 2020.

\bibitem[EHKS18]{conf/soda/EhsaniHKS18}
Soheil Ehsani, MohammadTaghi Hajiaghayi, Thomas Kesselheim, and Sahil Singla.
\newblock Prophet secretary for combinatorial auctions and matroids.
\newblock In {\em {SODA}}, pages 700--714. {SIAM}, 2018.

\bibitem[EHLM15]{esa/EsfandiariHLM15}
Hossein Esfandiari, MohammadTaghi Hajiaghayi, Vahid Liaghat, and Morteza
  Monemizadeh.
\newblock Prophet secretary.
\newblock In {\em {ESA}}, volume 9294 of {\em Lecture Notes in Computer
  Science}, pages 496--508. Springer, 2015.

\bibitem[GM66]{gilbert1966recognizing}
John~P Gilbert and Frederick Mosteller.
\newblock Recognizing the maximum of a sequence.
\newblock {\em Journal of the American Statistical Association},
  61(313):35--73, 1966.

\bibitem[HHSW21]{conf/sigecom/Hu0SW21}
Yihang Hu, Zhiyi Huang, Yiheng Shen, and Xiangning Wang.
\newblock Targeting makes sample efficiency in auction design.
\newblock In {\em {EC}}, pages 610--629. {ACM}, 2021.

\bibitem[HK82]{hill1982comparisons}
Theodore~P Hill and Robert~P Kertz.
\newblock Comparisons of stop rule and supremum expectations of iid random
  variables.
\newblock {\em The Annals of Probability}, pages 336--345, 1982.

\bibitem[HKS07]{aaai/HajiaghayiKS07}
Mohammad~Taghi Hajiaghayi, Robert~D. Kleinberg, and Tuomas Sandholm.
\newblock Automated online mechanism design and prophet inequalities.
\newblock In {\em {AAAI}}, pages 58--65. {AAAI} Press, 2007.

\bibitem[JMM{\etalchar{+}}03]{DBLP:journals/jacm/JainMMSV03}
Kamal Jain, Mohammad Mahdian, Evangelos Markakis, Amin Saberi, and Vijay~V.
  Vazirani.
\newblock Greedy facility location algorithms analyzed using dual fitting with
  factor-revealing {LP}.
\newblock {\em J. {ACM}}, 50(6):795--824, 2003.

\bibitem[Ker86]{kertz1986stop}
Robert~P Kertz.
\newblock Stop rule and supremum expectations of iid random variables: a
  complete comparison by conjugate duality.
\newblock {\em Journal of multivariate analysis}, 19(1):88--112, 1986.

\bibitem[KS77]{krengel1977semiamarts}
Ulrich Krengel and Louis Sucheston.
\newblock Semiamarts and finite values.
\newblock {\em Bulletin of the American Mathematical Society}, 83(4):745--747,
  1977.

\bibitem[KS78]{krengel1978semiamarts}
Ulrich Krengel and Louis Sucheston.
\newblock On semiamarts, amarts, and processes with finite value.
\newblock {\em Probability on Banach spaces}, 4:197--266, 1978.

\bibitem[KW12]{stoc/KleinbergW12}
Robert Kleinberg and S.~Matthew Weinberg.
\newblock Matroid prophet inequalities.
\newblock In {\em {STOC}}, pages 123--136. {ACM}, 2012.

\bibitem[LSTW21]{corr/abs-2111-03158}
Renato~Paes Leme, Balasubramanian Sivan, Yifeng Teng, and Pratik Worah.
\newblock Pricing query complexity of revenue maximization.
\newblock {\em CoRR}, abs/2111.03158, 2021.

\bibitem[Luc17]{journals/sigecom/Lucier17}
Brendan Lucier.
\newblock An economic view of prophet inequalities.
\newblock {\em SIGecom Exch.}, 16(1):24--47, 2017.

\bibitem[MY11]{DBLP:conf/stoc/MahdianY11}
Mohammad Mahdian and Qiqi Yan.
\newblock Online bipartite matching with random arrivals: an approach based on
  strongly factor-revealing lps.
\newblock In Lance Fortnow and Salil~P. Vadhan, editors, {\em Proceedings of
  the 43rd {ACM} Symposium on Theory of Computing, {STOC} 2011, San Jose, CA,
  USA, 6-8 June 2011}, pages 597--606. {ACM}, 2011.

\bibitem[NV22]{DBLP:journals/corr/abs-2208-09159}
Pranav Nuti and Jan Vondr{\'{a}}k.
\newblock Secretary problems: The power of a single sample.
\newblock {\em CoRR}, abs/2208.09159, 2022.

\bibitem[PST22]{DBLP:journals/corr/abs-2210-05634}
Sebastian Perez{-}Salazar, Mohit Singh, and Alejandro Toriello.
\newblock The {IID} prophet inequality with limited flexibility.
\newblock {\em CoRR}, abs/2210.05634, 2022.

\bibitem[PT22]{DBLP:conf/focs/PengT22}
Bo~Peng and Zhihao~Gavin Tang.
\newblock Order selection prophet inequality: From threshold optimization to
  arrival time design.
\newblock In {\em 63rd {IEEE} Annual Symposium on Foundations of Computer
  Science, {FOCS} 2022, Denver, CO, USA, October 31 - November 3, 2022}, pages
  171--178. {IEEE}, 2022.

\bibitem[RS17]{conf/soda/RubinsteinS17}
Aviad Rubinstein and Sahil Singla.
\newblock Combinatorial prophet inequalities.
\newblock In {\em {SODA}}, pages 1671--1687. {SIAM}, 2017.

\bibitem[RWW20]{innovations/RubinsteinWW20}
Aviad Rubinstein, Jack~Z. Wang, and S.~Matthew Weinberg.
\newblock Optimal single-choice prophet inequalities from samples.
\newblock In {\em {ITCS}}, volume 151 of {\em LIPIcs}, pages 60:1--60:10.
  Schloss Dagstuhl - Leibniz-Zentrum f{\"{u}}r Informatik, 2020.

\bibitem[SC84]{samuel1984comparison}
Ester Samuel-Cahn.
\newblock Comparison of threshold stop rules and maximum for independent
  nonnegative random variables.
\newblock {\em the Annals of Probability}, pages 1213--1216, 1984.

\end{thebibliography}
 	\bibliographystyle{alpha}
}

\newpage

\appendix

\section{Justification of No Mass Point Assumption} \label{appendix:justification}

We provide a justification for the no-mass-point assumption by showing that when the distribution contains mass points, every single-threshold algorithm that makes an arbitrary number of queries to the oracle performs arbitrarily badly.
Consider an algorithm that makes queries at quantiles $0 < q_1 < \cdots < q_k < 1$ and sets a threshold $\theta$ for accepting variables.
Note that $\theta$ does not have to be equal to any of $v(q_i)$'s.
Now suppose that for all queries $\{q_i\}_{i\in[k]}$ the returned values are $1$, i.e., $v(q_i) = 1$ for all $i\in[k]$.
We consider how the algorithm sets the threshold.

If $\theta > 1$ then for the distribution $D$ such that $x = 1$ with probability $1$, we have $\ALG = 0$ because no variable will be accepted, and $\OPT = 1$.
Thus the performance is arbitrarily bad.
On the other hand, if $\theta \leq 1$ then we consider the following distribution.
Let $M \gg \frac{n}{1 - q_k^n}$ be an arbitrarily large number.
Let $x = 1$ with probability $q_k$ and $x = M$ otherwise.
Since $\theta \leq 1$, we have $\ALG = 1$ while $\OPT = 1\cdot q_k^n + M \cdot (1-q_k^n) \gg n$.
Again, the performance of the algorithm is arbitrarily bad.

\smallskip

Note that it is possible that the algorithm makes an additional query at $q_{k+1} = 1$.
However, this does not change the result because we can modify the first distribution slightly such that $x = 1$ with probability $1-\frac{1}{M^3}$ and $x = M$ otherwise.
Note that for such a distribution we have $\ALG \approx 0$ and $\OPT \approx 1$.
Then for both distributions the values returned by the oracle are exactly the same: $v(q_i) = 1$ for all $i\in [k]$ and $v(q_{k+1}) = M$.
However, there is no way to distinguish these two distributions when setting the threshold.

\section{Related Works} 
\label{sec:related-works}

\paragraph{Prophet Inequality with Complete Information.}
The theory of prophet inequality is initiated by  \cite{gilbert1966recognizing} in the late sixties, and has been widely studied in seventies and eighties; see, e.g.,  
\cite{krengel1977semiamarts,krengel1978semiamarts,samuel1984comparison,kertz1986stop,hill1982comparisons}.
The problem regained significant  interest in the last decade partly because of its application in posted price mechanisms that are widely adopted in (online) auctions \cite{aaai/HajiaghayiKS07,stoc/ChawlaHMS10,journals/sigecom/Lucier17}. 
For general (non-identical) distributions, the optimal competitive ratio is $0.5$ when the ordering of the variables is adversarial~\cite{krengel1977semiamarts,krengel1978semiamarts,samuel1984comparison}.
In contrast, Chawla et al. \cite{stoc/ChawlaHMS10} showed that when the algorithm can decide the arrival order of the variables, the competitive ratio can be improved to $1-1/e$. The ratio was further improved to $0.634$ and $0.669$ by  \cite{sigecom/AzarCK18} and \cite{soda/CorreaSZ19}, respectively.
These results also apply to the setting with random arrival orders, which is called the {\em prophet secretary problem} by \cite{esa/EsfandiariHLM15}.
{Very recently, Peng and Tang \cite{DBLP:conf/focs/PengT22} proposed a $0.7251$-competitive algorithm that selects the arrival order of the variables based on a reduction to a continuous arrival time design problem.}
The ratio is further improved to $0.7258$ by \cite{DBLP:journals/corr/abs-2211-04145}, who also provided a $0.7254$ upper bound for the competitive ratio of algorithms for the prophet secretary problem, separating the best possible ratios for the two problems.
Beyond the classic single-choice setting, recent works  studied the setting where multiple variables can be selected subject to some combinatorial constraints 
~\cite{stoc/ChawlaHMS10,stoc/KleinbergW12,esa/DuttingK15,geb/AzarKW19,siamcomp/DuttingFKL20}.
In the aforementioned works, the objective is additive over the selected variables and 
similar problems with 
non-additive objectives 
were investigated by \cite{conf/soda/RubinsteinS17} and  \cite{focs/DuttingKL20}.

\paragraph{Prophet Inequality with Sample Queries.}
With unknown i.i.d. distributions,   Correa et al. \cite{ec/CorreaDFS19} showed that the best an algorithm can do is $1/e$-competitive, and the competitive ratio cannot be improved even if it can observe $o(n)$ random samples. 
However, if the algorithm has $n$ samples, it can achieve a competitive ratio $\alpha \in [0.6321,0.6931]$.
If $O(n^2)$ samples are given, the competitive ratio can be further improved to $(0.7451-\epsilon)$, the best possible ratio even if the full distribution is known to the algorithm.
Later, Rubinstein et al. \cite{innovations/RubinsteinWW20} improved this result by showing that $O(n)$ samples suffice to achieve the optimal ratio.
With non-i.i.d. distributions,  Rubinstein et al. 
 \cite{innovations/RubinsteinWW20} proved that one random sample from each distribution is enough to define a $0.5$-competitive algorithm.
Nuti et al. \cite{DBLP:journals/corr/abs-2208-09159} studied the prophet secretary version of the same problem and showed that the probability of selecting the maximum value is $0.25$, which is  optimal.
For multi-choice prophet inequality problems with random samples, constant competitive algorithms have been analyzed  \cite{geb/AzarKW19,corr/abs-2111-03174}.
Beyond random samples, inaccurate prior distributions were considered by  \cite{ec/DuttingK19}, where the algorithm knows some estimations  of the true distributions under various metrics.
A simultaneous and independent work by  \cite{DBLP:journals/corr/abs-2210-05634} also considered the multi-threshold algorithms and proved that the competitive ratio improves with more thresholds and approaches $0.7451$.

\section{Improving the Ratio with Three or More Thresholds}
\label{sec:3-or-more-queries}

We show in this section that the two-threshold algorithm and its analysis can be naturally extended to multiple thresholds, and better competitive ratios can be achieved.
In general, the algorithm has $k$ thresholds, and divides the time horizon $\{1,\ldots,n\}$ into $k$ phases, where $k\geq 2$.

\paragraph{Algorithm.}
The algorithm has positive parameters $\{c_\ell,\rho_\ell\}_{1\leq \ell \leq k}$, where $c_1 < \cdots < c_k$ and $\sum_{\ell=1}^k \rho_\ell = 1$.
Each parameter $c_\ell$ corresponds to a threshold $\theta_\ell := v(1-c_\ell/n)$ via a query to the oracle.
The parameters $\rho_1,\ldots,\rho_k$ divide the time horizon into $k$ phases: phase $\ell$ includes variables $x_i$ satisfying $n\cdot \sum_{j< \ell} \rho_j < i \leq n\cdot \sum_{j\leq \ell} \rho_j$, and threshold $\theta_\ell$ is used to decide the acceptance of these variables.
For sufficiently large $n$, we can assume w.l.o.g. that $\rho_{\ell} \cdot n$ is an integer for all $\ell$.
Note that each phase $\ell$ contains exactly $\rho_\ell\cdot n$ variables.
We describe the algorithm formally in Algorithm~\ref{alg:k-threshold}.

\begin{algorithm}
\caption{$k$-Threshold Algorithm}\label{alg:k-threshold}
	\textbf{Input:} {{Number of items} $n$, {quantile choices} $\{c_\ell\}_{1\leq \ell \leq k}$ with $c_1 < c_2 < \cdots < c_k$, and {observation portions} $\{\rho_\ell\}_{1\leq \ell \leq k}$ with $\sum_{\ell=1}^k \rho_\ell = 1$.}
 
    \textbf{Output:} A variable selected by the gambler.
    
	\For{$\ell = 1,2,\ldots, k$}{
	   Set $\theta_{\ell} := v(1-c_{\ell}/n)$ by making a query to the oracle; \quad\quad\quad {\color{gray}// phase $\ell$} \\
		\For{$i=n\cdot \sum_{j< \ell} \rho_j+1,\ldots,n\cdot \sum_{j\le \ell} \rho_j$}{
		   \If{$x_i \geq \theta_\ell$}{ 
		          \textbf{return} $x_i$ and \textbf{terminate};
                }
        }
    }	   
\end{algorithm}

We call Algorithm~\ref{alg:k-threshold} the {\em $k$-threshold algorithm} and refer to the $\ell$-th for-loop as \emph{phase} $\ell$ of the algorithm, during which variables in phase $\ell$ are evaluated.
Note that $\theta_{\ell}$'s forming a  decreasing sequence of thresholds since $c_{\ell}$'s are increasing in $\ell$. 

In the following sections, we extend our analysis in Section~\ref{sec:warmup} and provide a general framework that gives the competitive ratio of the $k$-threshold algorithm. 
We first obtain a lower bound in terms of $\{\theta_{\ell}, \Delta(c_{\ell})\}_{1\le \ell \le k}$ for $\ALG$ in Section~\ref{ssec:alg-lb} in a similar way as in Section~\ref{ssec:alg-lb-2-thresholds}.
Then we provide several upper bounds for $\OPT$ in Section~\ref{ssec:opt-ub} by extending our analysis in Section~\ref{ssec:opt-ub-2-thresholds}.
We ensure that all of these bounds are linear in terms of $\{\theta_{\ell}, \Delta(c_{\ell})\}_{1\le \ell \le k}$, which enables us to construct an LP in Section~\ref{ssec:comp-ratio} whose optimal objective value lower bounds the competitive ratio.

\subsection{Lower Bounding $\ALG$}\label{ssec:alg-lb}

We derive a lower bound for the expected gain of the $k$-threshold algorithm in the following lemma. 

\begin{lemma} \label{lem:ALG-lb}
Let $\{c_{\ell},\rho_{\ell}\}_{1\le \ell \le k}$ be the positive parameters of a $k$-threshold algorithm and $\epsilon > 0$ be arbitrarily small.
For sufficiently large $n$, we have
\begin{equation}\label{eq:alg-lb}
   \ALG \ge (1-\epsilon)\cdot \sum_{\ell = 1}^k e^{-\sum_{z=1}^{\ell-1} \rho_z \cdot c_z} \cdot (1-e^{-\rho_{\ell} \cdot c_{\ell}}) \cdot \left(\theta_{\ell} +\frac{\Delta(c_{\ell})}{c_{\ell}}\right).
\end{equation}
\end{lemma}
\begin{proof}
Let $\pi \in [k+1]$ denote the {\em stopping phase} of the above algorithm. 
For each $\ell\in [k]$, we let $\tau_{\ell}$ be the stopping time {\em within} phase $\pi$: when $\pi = \ell$, we have $\tau_\ell \in [\rho_{\ell}\cdot n]$; otherwise $\tau_\ell$ is undefined.
We write the  expected gain of the algorithm as
\begin{equation*}
    \ALG=\sum_{\ell=1}^{k}\sum_{i=1}^{\rho_{\ell}\cdot n} \Pr[\pi \ge \ell]\cdot  \Pr[\tau_{\ell} = i \wedge  \pi = \ell \mid \pi \geq \ell] \cdot \bbE[x_i \mid \tau_{\ell} = i \wedge  \pi = \ell].
\end{equation*}

The algorithm accepts the $i$-th variable in phase $\ell$ if and only if (i) all variables from phase $z < \ell$ are less than their corresponding thresholds $\theta_z$'s; (ii) the first $(i-1)$ variables in phase $\ell$ are below $\theta_{\ell}$; and (iii) $x_i$ is at least $\theta_{\ell}$. 
Therefore, we have 
\begin{align*}
   \Pr[\pi \ge \ell]\cdot  \Pr[\tau_{\ell} = i \wedge  \pi = \ell \mid \pi \geq \ell] = \prod_{z=1}^{\ell-1} F(\theta_z)^{\rho_z\cdot n}\cdot  F(\theta_{\ell})^{i-1}\cdot G(\theta_{\ell}).
\end{align*}
Similar to Equation~\eqref{eq:alg-gain}, the expected gain of the algorithm given that it accepts $x_i$ is
\begin{equation*}
    \bbE[x_i \mid \tau_{\ell} = i \wedge  \pi = \ell]  = \theta_{\ell} +\frac{\Delta(c_{\ell})/n}{G(\theta_{\ell})}.
\end{equation*}
In summary, the expected gain of the algorithm is 
\begin{align*}
    \ALG 
    &= \sum_{\ell=1}^{k}\sum_{i=1}^{\rho_{\ell}\cdot n} \prod_{z=1}^{\ell-1} F(\theta_z)^{\rho_z\cdot n}\cdot  F(\theta_{\ell})^{i-1}\cdot G(\theta_{\ell}) \cdot \left(\theta_{\ell} +\frac{\Delta(c_{\ell})/n}{G(\theta_{\ell})}\right)  \\
    &= \sum_{\ell=1}^{k} \prod_{z=1}^{\ell-1} F(\theta_z)^{\rho_z\cdot n}\sum_{i=1}^{\rho_{\ell}\cdot n}  F(\theta_{\ell})^{i-1}\cdot \left(G(\theta_{\ell}) \cdot  \theta_{\ell} +\frac{\Delta(c_{\ell})}{n}\right).
\end{align*}

Using $F(\theta_z) = 1-\frac{c_z}{n}$ and for sufficiently large $n$ we have
\begin{align}
    \ALG 
    &= \sum_{\ell=1}^{k}\prod_{z=1}^{\ell-1}\left(1-\frac{c_z}{n}\right)^{\rho_z \cdot n} \sum_{i=1}^{\rho_{\ell}\cdot n}  \left(1-\frac{c_{\ell}}{n}\right)^{i-1}\cdot \left(\frac{c_{\ell}}{n}\cdot \theta_{\ell}+\frac{\Delta(c_{\ell})}{n}\right)\nonumber\\
    &\ge \sum_{\ell = 1}^k \prod_{z=1}^{\ell-1} (1-\epsilon)\cdot e^{-\rho_z \cdot c_z} \cdot \frac{1-(1-\frac{c_{\ell}}{n})^{\rho_{\ell} \cdot n}}{c_{\ell}/n}\cdot \left(\frac{c_{\ell}}{n}\cdot \theta_{\ell}+\frac{\Delta(c_{\ell})}{n}\right)\nonumber\\
    &\ge (1-\epsilon)\cdot \sum_{\ell = 1}^k e^{-\sum_{z=1}^{\ell-1} \rho_z \cdot c_z} \cdot (1-e^{-\rho_{\ell} \cdot c_{\ell}}) \cdot \left(\theta_{\ell} +\frac{\Delta(c_{\ell})}{c_{\ell}}\right), \label{eq:alg-exp}
\end{align}
where $\epsilon>0$ is an arbitrarily small constant. 
\end{proof}

\subsection{Upper Bounding $\OPT$}\label{ssec:opt-ub}

Given threshold $\theta_{\ell}$ for $1\le \ell \le k$, we can reuse the argument in Equation~\eqref{eq:opt-upper} and obtain the following set of $k$ upper bounds for $\OPT$:
\begin{align*}
    \OPT &\le \theta_{\ell} + n\cdot \bbE[(x-\theta_{\ell})^+] = \theta_{\ell} + \Delta(c_{\ell}), \quad \forall 1 \le \ell \le k. 
\end{align*}

However, as discussed in Section~\ref{ssec:opt-ub-2-thresholds}, the above bounds alone would not be enough to beat $1-1/e$. 
In what follows, we give a more careful analysis that leads to tighter upper bounds for $\OPT$.
Recall that we have $0 \leq \theta_k < \theta_{k-1} < \cdots < \theta_1$.
We introduce  new variables $\delta_{\ell}$'s for $1 \le \ell < k$,  
\begin{equation*}
\delta_{\ell} = \int_{\theta_{\ell+1}}^{\theta_\ell} \Pr[x^* \geq t] dt, \quad \forall 1 \le \ell < k.
\end{equation*}

Recall that $x^* = \max_{i\in[n]} \{x_i\}$.
We have the following upper bound for $\OPT$:
\begin{align}
    \OPT = \bbE[x^*] &= \int_{0}^{\infty} \Pr[x^* \ge t] dt \nonumber \\
    &= \int_{0}^{\theta_k}\Pr[x^* \ge t] dt  +  \sum_{\ell=1}^{k-1}\int_{\theta_{\ell+1}}^{\theta_{\ell}} \Pr[x^* \ge t] dt + \int_{\theta_{1}}^{\infty} \Pr[x^* \ge t] dt\nonumber\\
    &\leq \theta_k  +\sum_{\ell=1}^{k-1}\delta_{\ell} + \bbE[(x^* - \theta_1)^+]
    \leq \theta_k  +\sum_{\ell=1}^{k-1}\delta_{\ell}+\Delta(c_1). \label{eq:opt-constr}
\end{align}  

In the next three lemmas, we establish three upper bounds for $\delta_{\ell}$'s, which, when combined with Equation~\eqref{eq:opt-constr}, provide much better upper bounds on \OPT. 

\begin{lemma}\label{lem:delta-ub-1}
For arbitrarily small $\epsilon>0$ and sufficiently large $n$, we have 
\begin{equation}\label{eq:delta-ub-1}
    \delta_{\ell} \le (1+\epsilon)\cdot (\theta_{\ell}-\theta_{\ell+1})\cdot (1-e^{-c_{\ell+1}}), \quad \forall 1 \le \ell < k.
\end{equation}
\end{lemma}
\begin{proof}
For $1 \le \ell < k$, we have 
\begin{align*}
    \delta_{\ell} &= \int_{\theta_{\ell+1}}^{\theta_{\ell}} \Pr[x^* \ge t] dt \le   \Pr[x^* \ge \theta_{\ell+1}]\cdot (\theta_{\ell}-\theta_{\ell+1})\\
    &= \left(1-(1-\frac{c_{\ell+1}}{n})^n\right)\cdot (\theta_{\ell}-\theta_{\ell+1})
    \le  (1+\epsilon)\cdot(1-e^{-c_{\ell+1}})\cdot (\theta_{\ell}-\theta_{\ell+1}),
\end{align*}
where the second equality holds due to $\Pr_{x\sim  D}[x < \theta_{\ell+1}] = 1-c_{\ell+1}/n$. 
\end{proof}

The second upper bound for $\delta_{\ell}$'s can be achieved using the union bound. 

\begin{lemma}\label{lem:delta-ub-2}
	We have
	\begin{equation}\label{eq:delta-ub-2}
    	\delta_{\ell} \le \Delta(c_{\ell+1})-\Delta(c_{\ell}), \quad \forall 1 \le \ell < k.
	\end{equation}
\end{lemma}
\begin{proof}
For $1 \le \ell < k$, we have 
\begin{align*}
    \delta_{\ell}&=\int_{\theta_{\ell+1}}^{\theta_{\ell}} \Pr[x^* \ge t] dt =  \int_{\theta_{\ell+1}}^{\theta_{\ell}} \Pr[(x_1 \ge t) \vee (x_2 \ge t) \vee \ldots \vee (x_n\ge t)] dt \\
    &\le n\cdot \int_{\theta_{\ell+1}}^{\theta_{\ell}} \Pr[x\ge t] dt = n\cdot \int_{\theta_{\ell+1}}^{\infty} \Pr[x\ge t] dt - n \cdot \int_{\theta_{\ell}}^{\infty} \Pr[x\ge t] dt \\
    & = \Delta(c_{\ell+1})-\Delta(c_{\ell}),
\end{align*}
where the last equality uses $\int_{\theta_{\ell}}^{\infty} \Pr[x\ge t] dt = \bbE[(x - \theta_\ell)^+]$ and $\theta_\ell = v(1-c_\ell/n)$.
\end{proof}

Finally, we derive the third upper bound on $\delta_l$'s.

\begin{lemma}\label{lem:delta-ub-3}
	For all $1 \le \ell < k$ and constant $c\in [c_{\ell},c_{\ell+1}]$, let $\gamma := \frac{c-c_\ell}{e^{-c} - e^{-c_{\ell+1}}}$.
	For arbitrarily small $\epsilon>0$ and sufficiently large $n$, we have 
	\begin{equation}\label{eq:delta-ub-3}
	\gamma \cdot \delta_{\ell}/(1+\epsilon) \le 
	\Delta(c_{\ell+1})-\Delta(c_{\ell}) - (c_{\ell} - \gamma(1-e^{-c}))\cdot (\theta_\ell - \theta_{\ell+1}).
	\end{equation}
\end{lemma}
\begin{proof}
    Recall $\Pr[x\geq \theta_{\ell}] = c_{\ell}/n$,  $\Pr[x\geq \theta_{\ell+1}] = c_{\ell+1}/n$ and that
    $\Pr[x\geq t]$ is strictly decreasing in $t$. 
    For all $c\in [c_{\ell},c_{\ell+1}]$, there must exist $t^* \in [\theta_{\ell+1}, \theta_{\ell}]$ such that $\Pr[x\geq t^*] = c/n$.
	Thus we have, 
	\begin{equation*}
	\Pr[x\geq t] \in 
	\begin{cases}
	[\frac{c}{n} , \frac{c_{\ell+1}}{n}], \quad & \forall t\in [\theta_{\ell+1}, t^*]; \\
	[\frac{c_{\ell}}{n} , \frac{c}{n}), \quad & \forall t\in (t^*, \theta_{\ell}].
	\end{cases}
	\end{equation*}
	Since $\Pr[x^* \geq t] = 1 - (1 - \Pr[x \geq t])^n$, 
	for sufficiently large $n$, we have:
	\begin{align*}
			\forall t\in [\theta_{\ell+1}, t^*],\quad & \Pr[x^* \geq t] \leq 1 - \left(1 - \frac{c_{\ell+1}}{n}\right)^n \leq (1+\epsilon)\cdot(1 - e^{-c_{\ell+1}}); \\
			\forall t\in (t^*, \theta_{\ell}], \ \quad  & \Pr[x^* \geq t] < 1 - \left(1 - \frac{c}{n}\right)^n \leq (1+\epsilon)\cdot(1 - e^{-c}).
	\end{align*}
 
	For convenience, we define $\alpha := \frac{t^* - \theta_{\ell+1}}{\theta_{\ell} - \theta_{\ell+1}} \in [0,1]$.
	We have
	\begin{align*}
	\delta_{\ell}=\int_{\theta_{\ell+1}}^{\theta_{\ell}} \Pr[x^* \geq t] dt 
	& = \int_{\theta_{\ell+1}}^{t^*} \Pr[x^* \geq t] dt + \int_{t^*}^{\theta_{\ell}} \Pr[x^* \geq t] dt \\
	& \leq (1+\epsilon)\cdot(t^* - \theta_{\ell+1})\cdot ( 1 - e^{-c_{\ell+1}} )+ (1+\epsilon)\cdot (\theta_{\ell} - t^*)\cdot ( 1 - e^{-c} ) \\
	& = (1+\epsilon)\cdot (\theta_{\ell} - \theta_{\ell+1})\cdot (1-e^{-c} + \alpha\cdot (e^{-c} - e^{-c_{\ell+1}})).
	\end{align*}
	Rearranging the inequality gives
	\begin{equation*}
	\alpha \geq \frac{1}{e^{-c} - e^{-c_{\ell+1}}}\cdot \left( \frac{\delta_{\ell}}{(1+\epsilon)(\theta_{\ell} - \theta_{\ell+1})} - (1-e^{-c}) \right).
	\end{equation*}

	On the other hand, we have (recall that we define $\gamma = \frac{c - c_{\ell}}{e^{-c} - e^{-c_{\ell+1}}}$)
	\begin{align*}
	n \cdot \int_{\theta_{\ell+1}}^{\theta_{\ell}} \Pr[x \geq t] dt 
	& = n \cdot\int_{\theta_{\ell+1}}^{t^*} \Pr[x \geq t] dt + n \cdot\int_{t^*}^{\theta_{\ell}} \Pr[x \geq t] dt \\
	& \geq (t^* - \theta_{\ell+1})\cdot c + (\theta_{\ell} - t^*)\cdot c_{\ell} \\
	& = (\theta_{\ell} - \theta_{\ell+1})\cdot (c_{\ell} + \alpha\cdot (c - c_{\ell})) \\
	& \geq (\theta_{\ell} - \theta_{\ell+1})\cdot \left\{c_{\ell} + \frac{c - c_{\ell}}{e^{-c} - e^{-c_{\ell+1}}}\cdot \left( \frac{\delta_{\ell}}{(1+\epsilon)(\theta_1 - \theta_{\ell+1})} - (1-e^{-c}) \right) \right\} \\
	& = c_{\ell} (\theta_{\ell} - \theta_{\ell+1}) + \frac{\gamma}{1+\epsilon} \cdot \delta_{\ell}- \gamma (1-e^{-c})\cdot (\theta_{\ell} - \theta_{\ell+1}).
	\end{align*}

	Recall that $n \cdot \int_{\theta_{\ell+1}}^{\theta_{\ell}} \Pr[x \geq t] dt  = \Delta(c_{\ell+1})-\Delta(c_{\ell})$. The above lower bound implies
	\begin{equation*}
	 \Delta(c_{\ell+1})-\Delta(c_{\ell}) \geq (c_{\ell} - \gamma(1-e^{-c}))\cdot (\theta_{\ell} - \theta_{\ell+1}) + \frac{\gamma}{1+\epsilon} \cdot \delta_{\ell}.
	\end{equation*}
	Rearranging the inequality completes the proof.
\end{proof}

\subsection{Competitive Ratios for $k$-Threshold Algorithms}\label{ssec:comp-ratio}

To recap, we have established a lower bound for $\ALG$ in Section~\ref{ssec:alg-lb} as well as several upper bounds for $\OPT$ in Section~\ref{ssec:opt-ub}.
Note that given  parameters $\{c_{\ell},\rho_{\ell}\}_{1\le \ell\le k}$, these bounds are linear in $\{\theta_{\ell}, \Delta(c_{\ell})\}_{1\le \ell \le k}$ and $\{\delta_{\ell}\}_{1\le \ell <k}$. 
We can now construct a minimization LP taking Equation~\eqref{eq:alg-lb} as the objective and Equations~\eqref{eq:opt-constr}, \eqref{eq:delta-ub-1}, \eqref{eq:delta-ub-2}, \eqref{eq:delta-ub-3} as constraints with non-negative LP variables $\{\theta_{\ell}, \Delta(c_{\ell})\}_{1\le \ell \le k}$ and $\{\delta_{\ell}\}_{1\le \ell <k}$. 
In the following theorem, we prove that the optimal objective of the LP gives a lower bound on the competitive ratio for the algorithm.

\begin{theorem}\label{thm:LP-lb}
    Given parameters $\{c_{\ell},\rho_{\ell}\}_{1\le \ell\le k}$ of a $k$-threshold algorithm, the optimal  value of the following LP provides a lower bound for the competitive ratio of the algorithm when $n\rightarrow \infty$.
    \begin{align}
		\text{minimize} \quad & \textstyle \sum_{\ell = 1}^k  e^{-\sum_{z=1}^{\ell-1} \rho_z \cdot c_z} \cdot (1-e^{-\rho_{\ell} \cdot c_{\ell}}) \cdot \left(\theta_{\ell} +\frac{\Delta(c_{\ell})}{c_{\ell}}\right) \nonumber\\
		\text{subject to} \quad 
		& 0 \le \theta_k \le \cdots \le \theta_1, \nonumber\\
		& \Delta(c_1),\delta_1, \ldots, \delta_k \ge 0,\nonumber\\
		&  \textstyle  1 \le \theta_k+\sum_{\ell=1}^{k-1} \delta_{\ell} + \Delta(c_1), \nonumber\\
		& \delta_{\ell} \le (\theta_{\ell}-\theta_{\ell+1})\cdot(1-e^{-c_{\ell+1}}), \qquad \forall \ell \in [k-1],\nonumber\\
		& \delta_{\ell} \le \Delta(c_{\ell+1})-\Delta(c_{\ell}), \qquad \forall \ell \in [k-1], \nonumber\\
		& \gamma \cdot \delta_{\ell} \le 
	\Delta(c_{\ell+1})-\Delta(c_{\ell}) - (c_{\ell} - \gamma(1-e^{-c}))\cdot (\theta_\ell - \theta_{\ell+1}),\nonumber\\
	&\qquad\qquad\qquad \forall \ell \in [k-1], \>\forall c \in [c_{\ell},c_{\ell+1}],  \> \textstyle \gamma = \frac{c-c_{\ell}}{e^{-c}-e^{-c_{\ell+1}}}. \nonumber
	\end{align}
\end{theorem}
\begin{proof}
Let $\{c_{\ell},\rho_{\ell}\}_{1\le \ell \le k}$ be the parameters of a $k$-threshold algorithm and $\epsilon > 0$ be arbitrarily small.
The objective of the LP comes from Equation~\eqref{eq:alg-lb}, the lower bound for $\ALG$, where we omit the $(1-\epsilon)$ term since we have $\epsilon \rightarrow 0$ when $n\rightarrow \infty$.
We will also omit the $(1+\epsilon)$ terms in the following for the same reason.
The first two constraints follow straightforwardly from the definitions of the parameters.
By scaling we can assume w.l.o.g. that $\OPT=1$.
Therefore Equation~\eqref{eq:opt-constr} gives the third constraint.
The three proceeding sets of constraints follow from Equations~\eqref{eq:delta-ub-1}, \eqref{eq:delta-ub-2} and \eqref{eq:delta-ub-3}, the three upper bounds on $\delta_{\ell}$'s that we have established in Section~\ref{ssec:opt-ub}.

Observe that every distribution $D$ induces a set of LP variables $\{\theta_{\ell}, \Delta(c_{\ell})\}_{1\le \ell \le k}$ as well as $\{\delta_{\ell}\}_{1\le \ell <k}$.
Moreover, by our analyses, they form a feasible solution to the LP.
Since the objective of any feasible solution induced by distribution $D$ provides a lower bound on $\ALG/\OPT$, 
the optimal (minimum) objective of the LP provides a lower bound on the competitive ratio, i.e., the worst-case performance against all distributions. 
\end{proof}

With the above result, it remains to set the appropriate parameters of the algorithm such that the optimal value of the LP is as large as possible.
We demonstrate the effectiveness of our LP framework with the following theorem that gives lower bounds of the competitive ratio for the $k$-threshold algorithm with $k\in \{2,3,4,5\}$.

\begin{theorem}
The $k$-threshold algorithm achieves a competitive ratio of at least $0.6786$ when $k=2$; $0.6883$ when $k=3$; $0.6946$ when $k=4$; and $0.7004$ when $k=5$.
\end{theorem} 
\begin{proof}
For two thresholds, we have three parameters for the algorithm, $c_1$, $c_2$, and $\rho := \rho_1$ (note that $\rho_2 = 1-\rho_1 = 1-\rho$). 
The LP has four variables: $\theta_1$, $\theta_2$, $\Delta_1 := \Delta(c_1)$ and $\Delta_2 := \Delta(c_2)$, and is written as follows. 
\begin{equation*}
    \begin{array}{ll}
        \text{minimize} & (1-e^{-\rho c_1})\theta_1+\frac{1-e^{-\rho c_1}}{c_1}\Delta_1+e^{-\rho c_1}(1-e^{-(1-\rho) c_2})\theta_2+\frac{e^{-\rho c_1}(1-e^{-(1-\rho) c_2})}{c_2}\Delta_2, \\[4pt]
        \text{subject to} & \theta_1 \ge \theta_2 \ge 0, \> \Delta_1 \ge 0, \> \delta_1 \ge 0, \\[4pt]
        & 1 \le \theta_2+\delta_1+\Delta_1,  \\[4pt]
        & \delta_1 \le (\theta_1-\theta_2)\cdot (1-e^{-c_2}),\\[4pt]
        & \delta_1 \le \Delta_2-\Delta_1,\\[4pt]
        & \gamma\cdot \delta_1 \le \Delta_2-\Delta_1-(c_1-\gamma(1-e^{-c}))\cdot(\theta_1-\theta_2), \> \forall c \in [c_1,c_2], \> \gamma=\frac{c-c_1}{e^{-c}-e^{-c_2}}.\\[4pt]
    \end{array}
\end{equation*}

By Theorem~\ref{thm:LP-lb}, given values of $c_1$, $c_2$, $\rho$, the objective value of the above LP bounds the competitive ratio from below.
Using an online LP solver, it can be verified that by fixing $c_1=0.7067$, $c_2=1.8353$ and $\rho=0.6204$, the optimal value of the LP is at least $0.6786$.

\begin{table}[ht]
\centering
\begin{tabular}{ccccc}
\hline
     & 2 thresholds & 3 thresholds & 4 thresholds & 5 thresholds \\\hline
   $c_1$  & $0.7067$ & $0.7204$& $0.6857$& $0.6561$ \\
   $c_2$ & $1.8353$ & $1.7551$ & $1.4367$ & $1.4082$ \\
   $c_3$ & -- & $3.2857$ & $2.4417$& $2.2735$ \\
   $c_4$ & -- & -- & $3.9036$& $3.4423$ \\
   $c_5$ & -- & -- & -- &  $4.4783$ \\\hline
   $\rho_1$  & $0.6204$ & $0.71$ &$0.65$ & $0.64$\\
   $\rho_2$ & $0.3796$ & $0.195$ &$0.19$ & $0.17$ \\
   $\rho_3$ & -- & $0.095$ &$0.1$ & $0.11$ \\
   $\rho_4$ & -- & -- &$0.06$ & $0.06$ \\
   $\rho_5$ & -- & -- & -- & $0.02$ \\\hline
   {\bf ratio} & $\bm{0.6786}$ & $\bm{0.6883}$ & $\bm{0.6946}$ & $\bm{0.7004}$ \\\hline
 \end{tabular}
 \vspace{5pt}
 \caption{Parameters and the resulting competitive ratios of the $k$-threshold algorithms for various $k$.}
 \label{tab:comp-ratio-k>2}
 \vspace{-20pt}
 \end{table}

The proofs for $k\geq 3$ are identical but with different choices of parameters.
We summarize the choices of parameters we have found and the corresponding competitive ratios in Table~\ref{tab:comp-ratio-k>2}. 
\end{proof}

We know from  \cite{conf/soda/EhsaniHKS18} and   \cite{sigecom/CorreaFHOV17} that the optimal competitive ratios are $1-1/e$ and $0.7451$ with $1$ and $n$ thresholds, respectively.
Our results partially fill in the gap between these two results.
While we can continue the numerical computation to achieve better ratios for $k\geq 6$, we believe that a more meaningful open problem is to improve the current analyses as we discuss in  detail in Section~\ref{sec:open-problem}, and see what value of $k$ is sufficient to achieve a competitive ratio close to the upper bound of $0.7451$.

\section{Missing Proofs from Section \ref{sec:warmup}}
\label{app:missing-proofs}

\subsection{Proof of Lemma \ref{lem:delta-ub-1-2-thresholds}}

\begin{proof}
    By definition of $\delta$, we have 
    \begin{align*}
        \delta &= \int_{\theta_{2}}^{\theta_{1}} \Pr[x^* \ge t] dt \le   \Pr[x^* \ge \theta_{2}]\cdot (\theta_{1}-\theta_{2})\\
        &= \left(1-(1-\frac{c_{2}}{n})^n\right)\cdot (\theta_{1}-\theta_{2})
        \le  (1+\epsilon)\cdot(1-e^{-c_{2}})\cdot (\theta_{1}-\theta_{2}),
    \end{align*}
    where the last inequality holds for all $n = \omega(\frac{1}{\epsilon})$. 
\end{proof}

\subsection{Proof of Lemma \ref{lem:delta-ub-2-2-thresholds}}
\begin{proof}
    By definition of $\delta$, we have 
    \begin{align*}
        \delta & =\int_{\theta_{2}}^{\theta_{1}} \Pr[x^* \ge t] dt =  \int_{\theta_{2}}^{\theta_{1}} \Pr[(x_1 \ge t) \vee (x_2 \ge t) \vee \ldots \vee (x_n\ge t)] dt \\
        &\le n\cdot \int_{\theta_{2}}^{\theta_{1}} \Pr_{x\sim D}[x\ge t] dt = n\cdot \int_{\theta_{2}}^{\infty} \Pr_{x\sim D}[x\ge t] dt - n \cdot \int_{\theta_{1}}^{\infty} \Pr_{x\sim D}[x\ge t] dt 
        = \Delta(c_{2})-\Delta(c_{1}),
    \end{align*}
    where the inequality holds by using the union bound.
\end{proof}

\subsection{Proof of Lemma \ref{lem:delta-ub-3-2-thresholds}}

\begin{proof}
    Recall that $\Pr[x\geq \theta_{1}] = c_{1}/n$, $\Pr[x\geq \theta_{2}] = c_{2}/n$ and that $\Pr[x\geq t]$ is strictly decreasing in $t$. 
    For all $c\in [c_{1},c_{2}]$, there must exist $t^* \in [\theta_{2}, \theta_{1}]$ such that $\Pr[x\geq t^*] = c/n$.
	Thus we have, 
	\begin{equation*}
	\Pr[x\geq t] \in 
	\begin{cases}
	[\frac{c}{n} , \frac{c_{2}}{n}], \quad & \forall t\in [\theta_{2}, t^*]; \\
	[\frac{c_{1}}{n} , \frac{c}{n}), \quad & \forall t\in (t^*, \theta_{1}].
	\end{cases}
	\end{equation*}
	Since $\Pr[x^* \geq t] = 1 - (1 - \Pr[x \geq t])^n$, 
	for sufficiently large $n$, we have:
	\begin{align*}
		\forall t\in [\theta_{2}, t^*],\quad & \Pr[x^* \geq t] \leq 1 - \left(1 - \frac{c_{2}}{n}\right)^n \leq (1+\epsilon)\cdot(1 - e^{-c_{2}}); \\
		\forall t\in (t^*, \theta_{1}], \ \quad  & \Pr[x^* \geq t] < 1 - \left(1 - \frac{c}{n}\right)^n \leq (1+\epsilon)\cdot(1 - e^{-c}).
	\end{align*}
 
	For convenience, we define $\alpha := \frac{t^* - \theta_{2}}{\theta_{1} - \theta_{2}} \in [0,1]$.
	We have
	\begin{align*}
	\delta & = \int_{\theta_{2}}^{\theta_{1}} \Pr[x^* \geq t] dt 
	= \int_{\theta_{2}}^{t^*} \Pr[x^* \geq t] dt + \int_{t^*}^{\theta_{1}} \Pr[x^* \geq t] dt \\
	& \leq (1+\epsilon)\cdot(t^* - \theta_{2})\cdot ( 1 - e^{-c_{2}} )+ (1+\epsilon)\cdot (\theta_{1} - t^*)\cdot ( 1 - e^{-c} ) \\
	& = (1+\epsilon)\cdot (\theta_{1} - \theta_{2})\cdot (1-e^{-c} + \alpha\cdot (e^{-c} - e^{-c_{2}})).
	\end{align*}
	Rearranging the inequality gives
	\begin{equation*}
	\alpha \geq \frac{1}{e^{-c} - e^{-c_{2}}}\cdot \left( \frac{\delta}{(1+\epsilon)(\theta_{1} - \theta_{2})} - (1-e^{-c}) \right).
	\end{equation*}

	On the other hand, with the definition of $\gamma = \frac{c - c_{1}}{e^{-c} - e^{-c_{2}}}$, we have 
	\begin{align*}
	n \cdot \int_{\theta_{2}}^{\theta_{1}} \Pr[x \geq t] dt 
	& = n \cdot\int_{\theta_{2}}^{t^*} \Pr[x \geq t] dt + n \cdot\int_{t^*}^{\theta_{1}} \Pr[x \geq t] dt \\
	& \geq (t^* - \theta_{2})\cdot c + (\theta_{1} - t^*)\cdot c_{1} \\
	& = (\theta_{1} - \theta_{2})\cdot (c_{1} + \alpha\cdot (c - c_{1})) \\
	& \geq (\theta_{1} - \theta_{2})\cdot \left\{c_{1} + \frac{c - c_{1}}{e^{-c} - e^{-c_{2}}}\cdot \left( \frac{\delta}{(1+\epsilon)(\theta_1 - \theta_{2})} - (1-e^{-c}) \right) \right\} \\
	& = c_{1} (\theta_{1} - \theta_{2}) + \frac{\gamma}{1+\epsilon} \cdot \delta - \gamma (1-e^{-c})\cdot (\theta_{1} - \theta_{2}).
	\end{align*}

	Recall that $n \cdot \int_{\theta_{2}}^{\theta_{1}} \Pr[x \geq t] dt  = \Delta(c_{2})-\Delta(c_{1})$. The above lower bound implies
	\begin{equation*}
	 \Delta(c_{2})-\Delta(c_{1}) \geq (c_{1} - \gamma(1-e^{-c}))\cdot (\theta_{1} - \theta_{2}) + \frac{\gamma}{1+\epsilon} \cdot \delta.
	\end{equation*}
	Rearranging the inequality completes the proof.
\end{proof}

\subsection{Proof of Lemma \ref{lem:2-threshold:upper}}

\begin{proof}
Given parameters $c_1$, $c_2$, and $\rho$, following the proof of Lemma~\ref{lem:ALG-lb-2-thresholds}, when $n\to \infty$, the expected gain of the algorithm is
\begin{equation}\label{eq:ub-exp}
    \ALG = (1-e^{-\rho c_1})\theta_1+\frac{1-e^{-\rho c_1}}{c_1}\Delta_1+e^{-\rho c_1}(1-e^{-(1-\rho) c_2})\theta_2+\frac{e^{-\rho c_1}(1-e^{-(1-\rho) c_2})}{c_2}\Delta_2.
\end{equation}

Note that this is also the objective function of the factor revealing LP. 
Let $\epsilon>0$ be an arbitrarily small constant and consider the following three distributions.
\begin{enumerate}
    \item Let $M$ be a sufficiently large number such that $c_1,c_2 \gg 1/M$.
    Consider the distribution $D_1$ as follows. 
    With probability $\frac{1}{n\cdot M}$,  $x$ is uniformly distributed over $[M-\epsilon, M+\epsilon]$; with probability $1-\frac{1}{n\cdot M}$,  $x$ is uniformly distributed over $[0,\epsilon]$. 
    The CDF for $D_1$ is strictly increasing, and when $\epsilon \rightarrow 0$, $M\rightarrow \infty$, we have
    \begin{equation*}
		\OPT \approx M\cdot \left( 1 - (1-\frac{1}{n\cdot M})^n \right) \approx M\cdot \frac{1}{M} = 1. 
		\end{equation*}
    It is not hard to see that both $\theta_1$ and $\theta_2$ are at most $\epsilon$.
    Moreover, $\Delta_1=n\cdot \bbE[(x-\theta_1)^+] \approx n\cdot \bbE[x] = 1$ and similarly $\Delta_2\approx 1$. By plugging these values in Equation (\ref{eq:ub-exp}), under $n \rightarrow \infty$,  
    
    \begin{equation}
        \ALG \approx \frac{1-e^{-\rho c_1}}{c_1}+\frac{e^{-\rho c_1}(1-e^{-(1-\rho) c_2})}{c_2} := UB_1. \label{eq:2-thres-ub1}
    \end{equation}

    \item Consider the distribution $D_2$ such that $x$ is uniformly distributed over $[1,1+\epsilon]$. Note that $\OPT \approx 1$ as $\epsilon \rightarrow 0$ and the CDF is strictly increasing. 
    For this distribution, we have $\theta_1\approx\theta_2\approx1$ while $\Delta_1 \approx \Delta_2 \approx 0$. Therefore, by Equation (\ref{eq:ub-exp}) and when $n \rightarrow \infty$,
    \begin{equation}\label{eq:2-thres-ub2}
        \ALG \approx (1-e^{-\rho c_1})+e^{-\rho c_1}(1-e^{-(1-\rho)c_2}) = 1-e^{-\rho c_1-(1-\rho)c_2} := UB_2.
    \end{equation}
    
    \item Let $c\in (c_1,c_2)$. Consider the distribution $D_3$ as follows. 
    With probability $\frac{c}{n}$,  $x_i$ is uniformly distributed over $[\frac{1}{1-e^{-c}}-\epsilon,\frac{1}{1-e^{-c}}+\epsilon]$; with probability $1-\frac{c}{n}$,  $x_i$ is uniformly distributed over $[0,\epsilon]$.
    Note that the CDF for $D_3$ is strictly increasing.
    Meanwhile, when $n \rightarrow \infty$,
    \begin{equation*}
    \OPT \approx \left(1-(1-\frac{c}{n})^n\right)\cdot \frac{1}{1-e^{-c}} \approx (1-e^{-c})\cdot\frac{1}{1-e^{-c}} = 1.
    \end{equation*}
    
    Since $\frac{c_1}{n}<\frac{c}{n}<\frac{c_2}{n}$, we have  $\theta_1\approx \frac{1}{1-e^{-c}}$ and  $\theta_2 \approx 0$.
    Moreover, $\Delta_1 \approx 0$ and $\Delta_2 \approx n\cdot \bbE[x] \approx \frac{c}{1-e^{-c}}$.
    Again, using Equation (\ref{eq:ub-exp}), when $n\rightarrow \infty$ we have
    \begin{equation*}
        \ALG \approx \frac{1-e^{-\rho c_1}}{1-e^{-c}}+\frac{e^{-\rho c_1}(1-e^{-(1-\rho) c_2})}{c_2}\cdot \frac{c}{1-e^{-c}} := UB_3(c).
    \end{equation*}
    We define our third upper bound as the minimum over all possible $UB_3(c)$'s, i.e., 
    \begin{equation}\label{eq:2-thres-ub3}
        UB_3 := \min_{c\in (c_1,c_2)} \{ UB_3(c) \}. 
    \end{equation}
    \end{enumerate}

    It remains to show that when $\epsilon \rightarrow 0$, $n \rightarrow \infty$, and $M \rightarrow \infty$, given any values of $c_1,c_2$ and $\rho \in [0,1]$, at least one of Equations~\eqref{eq:2-thres-ub1}, \eqref{eq:2-thres-ub2}, and \eqref{eq:2-thres-ub3} is at most $0.7081$. 
    That is, 
    \begin{equation}\label{eq:2-thres-ub}
        \max_{c_1\in [0,n],c_2\in[0,n],\rho\in[0,1]}\left\{\min\{UB_1,UB_2,UB_3\}\right\} \le 0.7081.
    \end{equation}
    
    We first observe that if $c_1 \ge 2$, $UB_1$ is at most $1/2$. 
    Similarly, for $c_2 \ge 30$, $UB_1$ cannot be more than $1-1/e$.
    Therefore, we only need to consider the case where $c_1 \in [0,2]$ and $c_2 \in [0,30]$. 
    We can prove \eqref{eq:2-thres-ub} using computational tools by discretizing the domain for $c_1$, $c_2$, and $\rho$.
    Specifically, we replace the domains for $c_1$,   $c_2$ and $\rho$ with $I_{c_1} = \{\frac{2i}{10^5}: 1\le i \le 10^5\}$, $I_{c_2} = \{\frac{3i}{10^6}: 1\le i \le 10^6\}$ and $I_{\rho} = \{\frac{i}{10^5}: 1\le i \le 10^5\}$, respectively. 
    By enumerating the possible values of $c_1$, $c_2$ and $\rho$ within $I_{c_1}$, $I_{c_2}$ and $I_{\rho}$, we obtain an upper bound of $0.70804$, which is achieved by $c_1=0.51904$, $c_2=2.32059$ and $\rho=0.60473$. 
    It can be verified that discretizing the domain incurs an additive error of at most $2\times 10^5$. 
    Therefore, the upper bound under the continuous domain is at most $0.70806 < 0.7081$. 
\end{proof}

\end{document}